\documentclass[runningheads]{llncs}

\usepackage[utf8]{inputenc}
\usepackage{preamble}

\title{Cryptoeconomic Security for\\ Data Availability Committees}

\usepackage{etoolbox}
\newtoggle{fullversion}
\toggletrue{fullversion}   %% uncomment for the full version -- it also displays acknowledgements at the end.
\iftoggle{fullversion}
{
  \author{Ertem Nusret Tas \and Dan Boneh}
  \institute{Stanford University}
}
{
  \author{}
  \institute{}
}

\begin{document}

% Reduce the spacing around floats so that you do not need to manually set tons of \vspace.
\addtolength{\textfloatsep}{-0.2in} % spacing between floats and text
\addtolength{\floatsep}{-0.2in}     % spacing between two floats
\captionsetup[table]{font=small,skip=4pt}   % spacing between caption and figure
\captionsetup[figure]{font=small,skip=4pt}
\renewcommand\subsubsection{\@startsection{subsubsection}{3}{\z@}%
                           {-8\p@ \@plus -4\p@ \@minus -4\p@}% Formerly -18\p@ \@plus -4\p@ \@minus -4\p@
                           {-0.5em \@plus -0.22em \@minus -0.1em}%
                           {\normalfont\normalsize\bfseries\boldmath}}

\maketitle

\begin{abstract}
Layer 2 systems have received increasing attention due to their potential to scale the throughput of L1 blockchains. 
To avoid the cost of putting data on chain, these systems increasingly turn to off-chain data availability solutions such as data availability committees (DACs). 
However, placing trust on DACs conflicts with the goal of obtaining an L2 architecture whose security relies solely on the L1 chain.
To eliminate such trust assumptions, we propose a DAC protocol that provides financial incentives to deter the DAC nodes from adversarial behavior such as withholding data upon request.
We then analyze the interaction of rational DAC nodes and clients as a dynamic game, with a Byzantine adversary that can corrupt and bribe the participants. 
We also define a notion of optimality for the DAC protocols, inspired by fairness and economic feasibility.
Our main result shows that our protocol is optimal and guarantees security with the highest possible probability under reasonable assumptions on the adversary.
\end{abstract}

\section{Introduction}
\label{sec:introduction}
Layer~2 systems~\cite{survey-1,survey-2,rollup} are an important approach to scaling
the throughput of Layer~1 blockchains such as Ethereum.
One of the key challenges in securing an L2 system 
is {\em data availability}: 
how to ensure that the state of the L2 system 
is always available and can be
reconstructed when needed?
This data is needed to safely restart the L2 system after a failure,
and for basic operations such as deposits and withdrawals.
The data availability problem comes up in other contexts as well,
such as in decentralized storage systems~\cite{arweave,filecoin,permacoin}.

\noindent
There are three general approaches to data availability in L2 systems:
\begin{itemize}[topsep=1ex,itemsep=1ex]
\item {\em On-chain data:} 
Rollup systems~\cite{rollup} store all transaction data on a Layer~1 \emph{parent} chain, such as Ethereum.  
These systems rely on the security of the L1 nodes to ensure that the data is always available. 

\item {\em Off-chain data stored by a Data Availability Committee (DAC):} 
Other systems such as StarkEx~\cite{starkex}, zkPorter~\cite{zkporter} and EigenLayr~\cite{EigenLayer} use a DAC
to store data off-chain across a number of trusted nodes~\cite{da-offchain}.
While the DAC provides a gas-efficient alternative to on-chain data, 
these systems rely on the correct operation of the DAC nodes to ensure that the data remains available. 

\item {\em Off-chain data with (repeated) Data Availability Sampling (Celestium~\cite{celestium}):}
An enhancement to DACs employs data availability sampling~\cite{albassam2018fraud,codedmerkletree,KZGdas,danksharding}
so that light clients, such as rollup users, can identify unavailable blocks created by the DAC without attempting to download the full block.
This approach is being used by modular blockchains such as Celestia~\cite{lazyledger} and Polygon Avail~\cite{avail} 
that specialize in preserving other chains' data.
However, DAS does not remove the trust assumption placed on the DAC nodes for data availability, since it requires DAC members to reply to DAS queries for data recovery.
DAS also cannot ensure that data remains permanently available~\cite{protodanksharding}.
\end{itemize}
Providing a standalone data availability service, such as Celestia and others, 
reflects a general trend towards modularity in the design of blockchains.

\medskip
In this paper, we focus on the security of Data Availability Committees (DAC), 
namely the last two bullets on the previous page. 
A DAC consists of multiple DAC members, 
which we call {\em nodes}, 
that store copies of the data that should be made available
(\eg, data sent by the rollup sequencer).
These nodes are expected to provide the data to querying clients 
in a timely manner.
Since malicious DAC members can withhold the data, 
DACs typically replicate the data on each DAC node for fault tolerance.
Thus, as long as one member is honest, rollup clients would receive the data upon request.
Although the storage requirement of the DAC scales linearly in the number of nodes due to replication, this redundancy can be reduced through the use of erasure codes and polynomial commitments.
For instance, the semi-AVID-PR scheme~\cite{validium} uses linear erasure-correcting codes
and homomorphic vector commitments to guarantee data availability as long as over $2/3$ of the nodes faithfully follow the protocol.

A major drawback of DACs is the need to trust the DAC members.
Consider a compromised DAC, where the adversary can prevent the reconstruction of the data, for example, by controlling more than $1/3$ of the DAC members.
Such a DAC can evolve the rollup state using unavailable transaction data, and withhold this data from the rollup clients.
This prevents clients from issuing transactions, and enables the adversary to steal client funds through ransom attacks~\cite{ransom}.
Thus, using a DAC hinders the goal of realizing a trust-minimized scaling architecture 
that relies solely on the security of the L1 chain for the safety and liveness of the rollup\footnote{Although current rollup systems typically rely on a single honest sequencer to evolve the rollup state, as long as the rollup data is available (\eg, on the L1 chain), any rollup full node can step up to fulfill the sequencer's role if it fails.}.
Liveness signifies that the clients can submit new transactions to the rollup system, and the system processes these transactions.

Data availability sampling (DAS) does not improve the liveness guarantees over the basic DAC architecture.
If the DAC is not compromised, then DAS helps rollup clients verify that the rollup data is available without downloading all the data from the DAC.
However, if the DAC is compromised, DAS provides no guarantees for data availability.
The compromised DAC can update the rollup state with unavailable transactions, and ignore all DAS queries from the clients.
Hence, DAS needs the trust assumption placed on the DAC members for liveness.

\paragraph{\bf Incentive-based data availability.}
One way to strengthen the security of a DAC is to rely on
financial incentives to deter the DAC members from adversarial behavior such as withholding data and lazy validation, where the DAC members \emph{pretend} as if the data was stored.
There are solutions such as Proofs of Custody~\cite{proof-of-custody} using financial disincentives (\eg, slashing) to encourage the lazy DAC members to store the entrusted data.
However, as withholding data is not a provable offense, 
it not clear how to enforce the slashing of the adversarial members' stake when they do indeed store the data, yet refuse to reveal it upon request (even if DAS is being used).
Moreover, any incentive-based data availability proposal 
must be analyzed in the face of rational DAC members who may respond to bribes, and Byzantine adversaries who may offer bribes.

Our main contribution is a DAC protocol that introduces a \emph{slashing} mechanism for malicious DAC nodes that withhold data.
The bulk of the paper is a technical analysis of the protocol, and
proves its security under certain assumptions on the adversary's power.
Moreover, we show that our protocol is optimal in a rigorous sense. 
We define the security model and the optimality notions in Sections~\ref{sec:model} and~\ref{sec:analysis}. 

We model the interactions of the DAC as a dynamic game involving multiple parties:
\begin{itemize}[itemsep=1ex]
\item 
DAC members, \ie, \defn{nodes}, are denoted by 
$\node_1, \ldots, \node_N$,
where $N$ is the number of nodes.  
These nodes store the data provided by an external entity.

\item 
A \defn{client} $\client$ sends a sequence of data queries to the $N$ nodes.
Every node can either respond to $\client$ with the requested data, or not respond. 
We assume the data held by the nodes is signed by the data provider, 
so that integrity of the response is easily verified.
If a response contains incorrect data, it is treated as a non-response.

\item 
A \defn{contract} running on the L1 chain is used to resolve disputes
and punish misbehaving DAC nodes. 
In particular, all $N$ nodes are staked, and the stake is held
in the contract.
If the nodes do not respond to $\client$ with the requested data, $\client$ can send its query to the contract.
In this case, the nodes are obliged to post their responses to the contract.
If a node provably fails to do so,
the contract can slash that node by confiscating part of its stake. 
Part of the slashed stake is given to the client as compensation and the rest is burned.
The size of the per-node stake and the behavior of the contract are the key
design decisions for a DAC protocol. 
\end{itemize}
Nodes and clients are rational agents that seek to maximize their utilities.
An adversary $\adv$ who fully controls $f$ corrupt nodes may
try to bribe the remaining $N - f$ nodes to cause a client
query to fail.  
This will make the requested data unrecoverable. 
Our goal is to design a DAC, so that under reasonable assumptions
on the size of $f$ and on the adversary's budget,
every query from the client will succeed with probability 
at least $1-\epsilon$, for some small $\epsilon$. 

Queries from the client model data requests needed for normal operations such as withdrawals.
For instance, in a rollup system, clients might have to prove their account balances with respect to the latest state root, 
and they do so by presenting a Merkle proof for their account. 
A non-responsive DAC storing the latest state can delay withdrawals by refusing to provide these Merkle proofs.
In this case, each client can post a query to the contract, and force the nodes to place the requested proof on the L1 chain.
Our model for the DAC system and the incentivize mechanism enforced by the contract has applications beyond data availability, and can be used to incentivize the honest participation of nodes in any committee outside the L1 chain that provides a service (\eg Decentralized Oracle Networks~\cite{chainlink-1,chainlink-2}).
We discuss use cases for our DAC system in Section~\ref{sec:model}. 

\paragraph{\bf The DAC protocol.} 
Suppose every query requires at least $k$ nodes out of $N$ to respond either directly to the client, or to the contract, for the client to obtain an answer to its query.
If no erasure coding is used and the data is replicated across all nodes, then $k=1$, otherwise $k$ could be bigger than $1$.

\smallskip\noindent
The protocol proceeds in four steps:
\begin{itemize}[itemsep=1ex,topsep=1ex]
\item {\em step 1:}  the client $\client$ sends its query to all DAC nodes over the network.
\item {\em step 2:}  if $k$ or more nodes respond, then the client obtains the requested data and the protocol terminates.
\item {\em step 3:}  if by a certain timeout the client does not receive $k$ responses, it posts its query to the contract on chain. 
For this purpose, the client has to send a base payment to the contract, 
which is needed to deter spamming clients.
We discuss the choice of client payment amount in Section~\ref{sec:discussion}.

\item {\em step 4:}  all $N$ nodes are then asked to post their responses to the query on chain. 
The protocol terminates once a certain timeout is reached.  
\end{itemize}
It remains to describe what the contract does once the timeout is reached in step~4.  
Every node that does not post its response to the contract by the timeout loses part or all of its stake.
The precise \emph{slashing function} is explained in Section~\ref{sec:contract}.
Moreover, if by the timeout in step~4 the client does not obtain an answer to its query through the responses,
the client is compensated by the contract using the funds obtained from the slashed nodes.

\ignore{
The slashing amount is determined by the following rule:
if $k$ or more nodes posted a response (so that the client obtains an answer to its query)
then the slashing amount is $p_w + \epsilon$ for some small $\epsilon$.   
If fewer than $k$ nodes posted a response (so that the query is not answered)
then the non-responsive node loses its entire stake $p_s$.  
Moreover, in this case the client is compensated by $p_{\mathrm{comp}}$ coins.
This compensation is paid by the contract who obtains the funds by slashing the non-responsive nodes.
The rest of the slashed amounts are burned.
}

The question is how to analyze the security and performance of a contract
in comparison to other contracts. 
In Section~\ref{sec:analysis} we present four desirable properties that a slashing function should satisfy.
Informally, these properties are:
\begin{itemize}
    \item \emph{Symmetry.} Motivated by fairness, the slashing function does not depend on the identities of the nodes, only on their actions.
    \item \emph{No Reward.} The slashing function does not pay out any rewards to the responsive nodes.
    This is motivated by economic feasibility as the contract should maintain a non-negative balance, and discourage the nodes from forcing an on-chain interaction for extra payoff rather than answering over the network. (No rewards rule does not rule out flat rewards by other means.)
    \item \emph{Security Under No Attack.} The slashing function ensures that the client promptly learns the correct response to its query, if the adversary does not offer any bribes.
    This captures a minimal notion of security.
    \item \emph{Minimal Punishment.}
    The slashing function keeps the slashed amounts of non-responsive nodes at a minimum when the client obtains an answer to its query.
    Thus, when most nodes are responsive, those that fail to respond due to benign failures, \eg, crash faults, are not heavily penalized.
\end{itemize}
We then define a notion of optimality for these functions:
\begin{definition}[Informal]
A slashing function is \defn{optimal} with respect to a set of slashing functions ${\cal F}$,
if 
the function satisfies the following two conditions:
(i) upon sending its query to the contract, the client obtains an answer with the maximum probability from among all the functions in ${\cal F}$ given the worst adversary, and
(ii) when the client obtains an answer, the function imposes the minimal punishment on non-responsive nodes from among the functions in ${\cal F}$. 
\end{definition}
In Section~\ref{sec:analysis}, we show that our slashing function is optimal for both \emph{risk-neutral} and \emph{risk-averse} nodes 
among the set of all functions that satisfy the four desirable properties described above.
We also analyze the security of a dynamic game among a rational client and the DAC nodes.
We identify the conditions under which the client obtains an answer to its query without calling the contract.
The analysis of Section~\ref{sec:analysis} is the most technical part of the paper, and is our core contribution.

\paragraph{\bf Evaluation.}

In Section~\ref{sec:evaluation}, we evaluate the 
real-world performance of our optimal contract.
To match the number of Ethereum validators and the minimum value that can be staked as an independent validator on Ethereum, we set the total number of DAC nodes to $N = 300,000$ and the amount staked per node to $32$ ETH.
Then, given risk-neutral nodes, the adversary has to offer a total bribe of $\approx 3.2 \cdot 10^3$ ETH ($\approx 3.9$ million USD\footnote{Ethereum to USD conversion rate, $1 \text{ ETH } \approx 1231.0 \text{ USD}$, is the average Ethereum price on July 15, 2022~\cite{conversion-rate}.}) to the nodes, to reduce the security probability per query by a tiny amount, 
namely to reduce the probability that a client learns the answer to its query from $100\%$ to $99.9\%$.
To prevent clients from learning the answers over repeated queries, the adversary has to spend at least $3.9$ million USD {\em for each query}. 
As our contract is optimal, no other contract can force the adversary to pay a higher bribe for the same security probability.
The minimum bribe needed by the adversary to reduce the security probability increases as $N$ or the collateral grows, or as the nodes become more risk-averse.

\section{Model}
\label{sec:model}
\paragraph{Notation.}
We denote the security parameter by $\lambda$.
We say that an event happens with negligible probability, if its probability, as a function of $\lambda$, is $o(1/\lambda^{d})$ for all $d>0$. 
We say that an event happens with overwhelming probability if it happens except with probability negligible in $\lambda$.
If an event happens with probability $q+\mathrm{negl}(\lambda)$ or $q-\mathrm{negl}(\lambda)$, where $q$ is a non-negligible constant, for simplicity, we say that the event happens with probability $q$.
We assume that except with probability negligible in $\lambda$, the contract implements the specified slashing function correctly, the underlying cryptographic primitives are secure, and messages can be posted to the contract within bounded time.
We use the shorthand $[N]$ to denote the set $\{1,2,\ldots,N\}$.

\paragraph{Environment and the Adversary.}
Time is slotted, and the clocks of the client and nodes are synchronized\footnote{Bounded clock offsets can be captured by the network delay.}.
Messages, \eg, queries and replies, can only be sent at the beginning of a slot, and are delivered to the recipient by the end of the same slot by the environment $\mathcal{Z}$.

Adversary $\mathcal{A}$ is a probabilistic polynomial time algorithm.
Before the execution starts, $\adv$ corrupts $f$ nodes, which are subsequently called adversarial.
These nodes can deviate from the protocol arbitrarily (Byzantine faults) under $\mathcal{A}$'s control, which has access to their internal states.
The remaining $N-f$ nodes and the client are utility maximizing agents and can choose any action that gives them a higher utility.
In the subsequent analysis, we will assume that $\node_i$, $i = N-f+1,\ldots, N$ represent the adversarial nodes, and $f \leq N-k$.
Otherwise, it is impossible to guarantee the recovery of the answer to a query as the adversarial nodes can withhold their responses from the client and the contract.

Before the protocol execution starts, the adversary can also offer \emph{bribes} to the \emph{remaining} nodes and the client subject to constraints.
It has a supply of $p_0$ coins, which can be distribute to any subset of the nodes as additional payoff if the nodes adopt an adversarial action during the game.
Similarly, the adversary can give up to $p_1$ coins to the client if it adopts an adversarial action.
Such an adversary is called a $(p_0,p_1)$-adversary.
When the bribe offered to the client is irrelevant, we use the notation $p_0$-adversary.
($p_0$ and $p_1$ are adversary's resources that are \emph{beyond} the $f$ nodes corrupted by the adversary.)
Upon hearing an offer, each participant can independently choose to accept or reject the bribe depending on the expected utility.
Once a participant accepts the bribe, the adversary can monitor through the environment and the contract if the specified action was taken.
Although the action and the exchange of the bribe might not happen atomically, the adversary and nodes can ensure that no party deviates from its promise through a trusted third party, or repeated games 
(\cf Section~\ref{sec:analysis-repeated-games}).

\paragraph{Actions, Payoffs, and the Game.}
We next describe the dynamic game played by the client and the DAC nodes.
Before the game starts, the client $\client$ and the nodes are input a single query by the environment $\mathcal{Z}$.
Given a query, each node $\node_i$ can instantaneously generate a response $c_i$, called the \emph{clue}.
We assume that the correctness of these clues can be verified by the clients and the contract\footnote{For instance, correctness of the data shards in PoS Ethereum can be verified with respect to a KZG commitment on the blockchain~\cite{kate,protodanksharding,danksharding}.}.
The contract accepts a clue by a node if and only if it is the first correct response by the node to a query posted to the contract.
It records the time slots when each query or clue was received, in a contract state.
At the beginning of each slot, the participants learn about the state recorded at the end of the previous slot.

Let $p_s$ be the amount staked by a node to function as a DAC member.
It costs $p_c$ coins for the client to send a query to the contract, and $p_w$ coins for each node to prepare and post the corresponding clue to the contract. 
It is free to send a clue to the client over the network.
These parameters are summarized in Table~\ref{tbl:params}.
We assume that each node starts the game with a baseline payoff of $C=p_s+p_w$, as it has $p_s$ coins staked in the contract, and is assumed to have enough funds to post clues to the contract during the game\footnote{For risk-neutral nodes, the baseline is normalized to be $0$.}.
The client starts the game with an initial payoff of $0$.

The actions available to the client $\client$ and a node $\node$ at any slot $t$ are denoted as follows:
\begin{itemize}
    \item $\reply$: $\node$ sends a correct clue to $\client$ over the network at slot $t$.
    \item $\query$: $\client$ sends a query to the contract for the first time at slot $t$.
    \item $\place$: Replying to a query, $\node$ sends a correct clue to the contract for the first time at slot $t$.
\end{itemize}
The notation $\neg (.)$ is used to denote the opposite of the specified action.
At any time slot, a node can take an action $(a,b)$, where $a \in \{\reply,\neg\reply\}$ and $b \in \{\place,\neg\place\}$.
Similarly, the client can take an action from $\{\query,\neg\query\}$.
Although the clients and nodes can exchange messages other than queries and clues, only the queries, clues or their absence can lead to a change in their payoffs.
Since the participants play a dynamic game, the actions chosen at later slots can depend on the actions observed at the earlier ones.

The game ends, and the payoffs are realized at the beginning of slot $\Tanswer$.
If $\client$ finds out the correct answer to its query through the clues, either posted to the contract or sent over the network, by slot $\Tanswer$, it receives a payoff of $p_f$ coins.
We set $\Tanswer=4$ though it can be any sufficiently large constant.
In our model, $\Tanswer$ should be at least $4$ to guarantee any meaningful security.
The payoffs of the participants depend on the bribes $p_0$ and $p_1$, the collateral $p_s$, the variables $p_f, p_c$, $p_w$ selected by $\mathcal{Z}$, and the contract's \emph{slashing function}.

Utility of a participant is given as a function $U(.)$ of the payoff obtained at the end of the game.
In the subsequent sections, we will first consider risk-neutral nodes with a linear utility function $U(x)=x$, where $x$ is the net payoff at the end.
We will then analyze risk-averse nodes with a strictly concave utility function of the form $U(x)=(x)^\nu$, where $\nu \in (0,1)$.
We do not consider risk-seeking nodes with strictly convex utility functions, \eg, $U(x)=(x)^\nu$, $\nu > 1$, as such a function violates the law of diminishing marginal utility for the payoffs.

We will later also consider a sub-game that focuses exclusively on the interaction between the nodes and the contract.
In the game, a query appears in the contract at some slot $t$, and the nodes choose to post clues or not at slot $t+1$, after which the payoffs are realized.
These payoffs depend on the bribe $p_0$, the collateral $p_s$, the cost $p_w$, and the slashing function.

\paragraph{Security.}
We say $\Tanswer$-security is satisfied if the client receives $k$ or more correct clues from the nodes either over the network or through the contract by the \emph{beginning} of slot $\Tanswer$ with overwhelming probability.

\paragraph{Application.}
The game above models the withdrawal of client funds from a blockchain or rollup.
Each client has an account, represented as a key-value pair, and the balances of these accounts constitute the blockchain state.
The hashes of the key-value pairs are organized in a vector commitment, \eg, a sparse Merkle tree, with a constant size commitment, called the state root.
The state data is preserved by the DAC nodes and state commitments are posted to the chain.

To prove its account balance, a client requests a witness from the nodes for the inclusion of its account within the latest state.
If it does not receive a witness over the network, the client can complain on a \emph{smart contract} by sending a query that contains the hash of the account's key-value pair.
If the hash is a hiding commitment, the client can also ensure that no observer learns its balance.
It can always prove its balance to a select third party by revealing the key-value pair at the pre-image of the hash, the latest state root on chain, and the witness.

Upon receiving a query, the contract expects a witness to be provided by the DAC nodes within a bounded time, \eg, the chain's confirmation latency.
Correctness of this witness can be verified by the contract and the client with respect to the state commitment on the chain.
If the query is for an account not included in the latest state, the nodes can convince the contract of this fact via a proof of non-inclusion.
If there are multiple queries, instead of sending the witness for each query, the nodes can compute a SNARK proof that verifies the inclusion of all the queried accounts within the state.
Clients can then verify the inclusion of the queried accounts by checking the proof with respect to the latest state root, and the hashes of the queried accounts.
Succinctness of the SNARK proof enables achieving bounded delay on the response time.

\section{The Optimal Contract}
\label{sec:contract}
\vspace{-0.5cm}
\begin{table}
\begin{center}
\begin{tabular}{||c | c||} 
 \hline
 Parameter & Explanation \\ [0.5ex] 
 \hhline{||=|=||}
 $N$ & Number of nodes \\ 
 \hline
 $p_0$ & Total payoff the adversary can offer to the nodes \\ 
 \hline
 $p_1$ & Total payoff the adversary can offer to the client \\
 \hline
 $p_{\mathrm{comp}}$ & Compensation for the client if reconstruction fails \\
 \hline
 $p_f$ & Client's payoff from a valid reply within $4$ slots \\
 \hline
 $p_c$ & Cost of sending a query to the contract \\
 \hline
 $p_w$ & Cost of constructing and sending a clue to the contract \\
 \hline
 $p_s$ & Collateral per node\\
 \hline
\end{tabular}
\end{center}
\caption{Parameters in our model}
\vspace{-0.5cm}
\label{tbl:params}
\vspace{-0.2cm}
\end{table}

A contract can reward or punish the nodes depending on whether it received clues from the nodes for a query within a timeout period.
We normalize this timeout to be a single slot for all contracts\footnote{In a network with temporary partitions, the timeout can be increased to guarantee the timely inclusion of the messages sent to the contract.}.
Let $x_i=1$ if the node $\node_i$ sends a valid clue at slot $t+1$ in response to a query posted at some slot $t$, and $x_i=0$ otherwise.
We characterize a contract by a slashing function $f$ that maps actions $\mathbf{x} = (x_1,\ldots,x_N) \in \{0,1\}^N$ to payoffs $(f_1(\mathbf{x}), \ldots, f_N(\mathbf{x})) \in \mathbb{R}^N$ for the nodes, and the payoff $f_{\client}(\mathbf{x}) \in \mathbb{R}$ for the client.
%We also write $h_i$ for $f_i(\mathbf{x})$.
Since the contract cannot punish the nodes more than the staked collateral, $f_i(\mathbf{x}) \geq -p_s$ for every action $\mathbf{x} \in \{0,1\}^N$.
We will hereafter use slashing function and the contract interchangeably.

The proposed optimal contract and the associated slashing function is parameterized by a small number $\epsilon>0$:
\begin{align*}
f_i(\mathbf{x}) = & \begin{cases}
    0  & \text{ if } x_i = 1 \\
    -p_s  & \text{ if } \sum_{j=1}^N x_j < k \text{ and } x_i = 0 \\[1mm]
    -p_w-\epsilon  & \text{ if } \sum_{j=1}^N x_j \geq k \text{ and } x_i = 0
\end{cases}
\end{align*}
% \\[1em]
\begin{align*}
f_{\client}(\mathbf{x}) = & \begin{cases}
    0  & \text{ if } \sum_{j=1}^N x_j \geq k \\[1mm]
    p_{\mathrm{comp}}  & \text{ if } \sum_{j=1}^N x_j < k
\end{cases}
\end{align*}
Here, $p_{\mathrm{comp}}<p_s,p_f$, and $p_{\mathrm{comp}}>p_c$ to ensure that the client's net payoff stays above zero if it does not receive sufficiently many clues through the contract.

The contract burns, \ie, slashes the collateral $p_s$ put up by each node that has not sent a valid clue by the end of slot $t+1$, if there are less than $k$ clues.
In this case, the contract also awards $p_\mathrm{comp}$ of the slashed coins to $\client$.
Otherwise, if there are $k$ or more clues in the contract by slot $t+1$, it punishes the non-responsive nodes by a modest amount, namely $p_w+\epsilon$.

\section{Analysis}
\label{sec:analysis}

In Section~\ref{sec:analysis-contract-properties}, we formalize the desirable properties and notions of optimality for slashing functions.
In Section~\ref{sec:analysis-risk-neutral}, 
we show that the slashing function of Section~\ref{sec:contract} is optimal for risk-neutral and risk-averse nodes.
In Section~\ref{sec:analysis-game}, we generalize the analysis to a dynamic game with a rational client.
In Section~\ref{sec:analysis-repeated-games}, we analyze a repeated game played between the nodes and the adversary.

\subsection{Contract Properties}
\label{sec:analysis-contract-properties}

The desirable properties for a slashing function $f$, introduced in Section~\ref{sec:introduction}, are formalized below:
\begin{itemize}[itemsep=1ex]
    \item \emph{A1: Symmetry.} A slashing function $f$ is symmetric if $f(\pi(\mathbf{x})) = \pi(f(\mathbf{x}))$ for every action $\mathbf{x} \in \{0,1\}^N$ and permutation $\pi$.
    
    \item \emph{A2: No Reward.} A slashing function $f$ offers no rewards if for every action $\mathbf{x} \in \{0,1\}^N$, $f_i(\mathbf{x}) \leq 0$, $\forall i \in [N]$, and $f_{\client}(\mathbf{x}) + \sum_{i \in [N]} f_i(\mathbf{x}) \leq 0$.
    
    \item \emph{A3: Security Under No Attack.} A slashing function $f$ guarantees security under no attack if for all $(0,0)$-adversaries, it achieves $\Tanswer$-security with overwhelming probability in all Nash equilibria of the game.
    
    \item \emph{A4: $B$-Minimal punishment.}
    A slashing function $f$ offers $B$ minimal punishment if for every action $\mathbf{x} \in \{0,1\}^N$ such that $\sum_{i=1}^N x_i \geq k$, 
    we have that $f_i(\mathbf{x}) \geq -B$ for all $i \in [N]$.
\end{itemize}

\begin{definition}
\label{def:compliant}
A slashing function $f$ is said to be \defn{compliant} if it satisfies the axioms A1--A3, and the axiom A4 for some constant $B \in \mathbb{R}^+$.
\end{definition}

\begin{definition}
\label{def:tolerant}
A compliant slashing function $f$ is said to be \defn{$(p_0,q)$-tolerant} if for all $p_0$-adversaries, when a query is received by the contract at some slot $t$, there are $k$ or more correct clues in the contract at slot $t+1$, with probability at least $q$, in all Nash equilibria.
\end{definition}
The value $q$ of a $(p_0,q)$-tolerant contract can be interpreted as the minimum probability for security given that the client received no responses over the network and sent its query to the contract.

We next introduce two notions of optimality for the contract.
A security-optimal function ensures that for any $p_0$, security is violated with the minimum possible probability in the equilibrium with the largest failure probability.

\begin{definition}
\label{def:punishment-optimal}
A compliant slashing function $f$ is said to be \defn{security-optimal} if for all $p_0 \geq 0$, there exists a $q_0 \in [0,1]$ such that $f$ is $(p_0,q_0)$-tolerant, and there does not exist any compliant, $(p_0,q)$-tolerant function $f'$, where $q > q_0$.
\end{definition}

\noindent
A punishment-optimal contract imposes the minimum punishment on the unresponsive nodes (\eg, due to benign errors) if security was not compromised.
\begin{definition}
A compliant slashing function $f$ is said to be \defn{$\epsilon$-punishment-optimal} if it satisfies $B$-minimal punishment, and 
no compliant slashing function $f'$ can satisfy $B'$-minimal punishment for some $B'<B-\epsilon$.
\end{definition}

Finally, we combine the two notions of optimality in a single definition:
\begin{definition}
\label{def:optimal}
A family of slashing functions $f_\epsilon$, parameterized by $\epsilon$, is said to be \defn{optimal} if each member $f_\epsilon$ of the family is compliant, security-optimal and \emph{$\epsilon$-punishment-optimal}.
\end{definition}
\subsection{Analysis of The Optimal Contract}
\label{sec:analysis-risk-neutral}

We prove the following theorem for risk-neutral and risk-averse nodes.

\begin{theorem}
\label{thm:risk-neutral-optimality}
The family of slashing functions described in Section~\ref{sec:contract} is optimal.
\end{theorem}

Theorem~\ref{thm:risk-neutral-optimality} follows from Theorems~\ref{thm:axioms}, \ref{thm:axioms-converse}, and~\ref{thm:security-two-risk-averse}.
Their proofs for risk-neutral and risk-averse nodes are given in Appendices~\ref{sec:appendix-proofs-risk-neutral} and~\ref{sec:appendix-proofs-risk-averse} respectively.

We first showing that the slashing function is compliant:
\begin{theorem}
\label{thm:axioms}
Each slashing function from Section~\ref{sec:contract}, parameterized by $\epsilon>0$, 
satisfies symmetry (A1), no reward (A2), security under no attack (A3), and $(p_w+\epsilon)$-minimal punishment (A4).
\end{theorem}

The axioms A1, A2 and A4 follow by inspection, whereas A3 is shown by Lemma~\ref{lem:security-zero}.
Proof of Lemma~\ref{lem:security-zero} is given in Appendices~\ref{sec:appendix-proofs-risk-neutral} and~\ref{sec:appendix-proofs-risk-averse} for risk-neutral and risk-averse nodes respectively.

\begin{lemma}
\label{lem:security-zero}
Given the slashing function of Section~\ref{sec:contract}, for any $(0,0)$-adversary $\mathcal{A}$, $4$-security is satisfied with overwhelming probability in all Nash equilibria. 
\end{lemma}

When $p_{\mathrm{comp}}>p_c$, the client is incentivized to send its query to the contract if it receives less than $k$ clues over the network. 
Then, the nodes post their clues to the contract to avoid slashing of their stakes, and the contract ensures security with overwhelming probability.

\begin{remark}
If $p_\mathrm{comp} \leq p_c$, for any contract that offers no rewards to the nodes, and for any $(\epsilon,0)$-adversary where $\epsilon \geq 0$, there exists a Nash equilibrium such that $4$-security is violated with overwhelming probability.
Consider the action profile, where the nodes do not send their clues to the client $\client$ over the network, and do not post their clues to the contract.
Given these actions, if $p_\mathrm{comp} \leq p_c$, $\client$'s payoff can at most be $0$, and the maximum payoff is achieved if $\client$ does not send a query to the contract, even when it does not receive clues over the network.
In this case, the normalized payoff of each node becomes $0$ as well, which is the maximum payoff attainable by any node. 
Hence, the nodes do not have any incentive to deviate from the action profile above, which constitutes a Nash equilibrium.
\end{remark}

We next show that the slashing function is $\epsilon$-punishment optimal.
\begin{theorem}
\label{thm:axioms-converse}
Consider a slashing function that is symmetric (A1), offers no rewards (A2), and satisfies $B$-minimal punishment for some $B<p_w$ (A4).
Then, for $k>1$, there exists a $(0,0)$-adversary $\mathcal{A}$ and a Nash equilibrium, where $4$-security is violated with non-negligible probability.
Thus, no compliant slashing function can satisfy $B$-minimal punishment for some $B<p_w$.
\end{theorem}
When $B<p_w$, punishment for a node that does not post its clue to the contract while the other nodes send their clues is smaller than the cost of posting the clue.
This leads to a free-rider problem, and results in an equilibrium with a non-negligible failure probability for security, where each non-adversarial node trusts the others to send clues to the contract.

Finally, we prove security-optimality:
\begin{theorem}
\label{thm:security-two-risk-averse}
The slashing function of Section~\ref{sec:contract} is security optimal.
\end{theorem}

Consider the sub-game, where the contract receives a query at some slot $t$.
For a given contract and utility function $U(x)=x^\nu$, let $q^{\mathcal{A}}_{\mathrm{v}}$ denote the probability that given a $p_0$-adversary $\mathcal{A}$, there are less than $k$ valid clues in the contract at slot $t+1$ in the Nash equilibrium with the largest probability of failure. 
Then, the proof of Theorem~\ref{thm:security-two-risk-averse} for risk-neutral nodes follow from Theorem~\ref{thm:security-two}:
\begin{theorem}
\label{thm:security-two}
Suppose $p_0 < (N-f-k+1)(p_s-p_w)$ and the nodes are risk-neutral with the utility function $U(x)=x$.
Then, for any $p_0$-adversary $\mathcal{A}$, the slashing function of Section~\ref{sec:contract} satisfies
\begin{IEEEeqnarray*}{C}
q^{\mathcal{A}}_{\mathrm{v}} \leq q^* = \frac{p_0}{(N-f-k+1)(p_s-p_w)}
\end{IEEEeqnarray*}
Moreover, there exists a $p_0$-adversary $\mathcal{A}$ such that for any compliant slashing function, $q^{\mathcal{A}}_{\mathrm{v}} \geq q^*$.
\end{theorem}
The $p_0$-adversary $\mathcal{A}$ of Theorem~\ref{thm:security-two} offers a bribe of $\frac{p_0}{N-f-k+1}$ to $N-f-k+1$ non-adversarial nodes, \eg, $\node_i$, $i \in [N-f-k+1]$.
In return, it requests these nodes to collectively withhold their clues from the contract with probability $q^*$.

\begin{remark}
\label{remark:strong-adversary}
If $p_0 \geq (N-f-k+1)(p_s-p_w)$, there exists a Nash equilibrium, where $4$-security is violated with overwhelming probability. 
Adversary offers a payoff of $p_s-p_w$ to each of the $N-f-k+1$ nodes, and requests them to withhold their clues from the contract.
In the equilibrium, the offer is accepted and the nodes do not post their clues to the contract.
\end{remark}

\begin{remark}
Sending repeated queries to the contract does not reduce the failure probability by more than a linear factor in latency. 
Suppose the client $\client$ is allowed to send the same query to the contract up to $\ell$ times.
Then, if there are less than $k$ valid clues in the contract at slot $t+1$, $\client$ might want to repeat the sub-game up to $\ell$ times with the hope of eventually learning the answer to its query.
In this case, the adversary $\mathcal{A}$ can offer a payoff of $\frac{p_0}{N-f-k+1}$ to the nodes $\node_i$, $i \in [N-f-k+1]$, and in return, ask them to collectively withhold their clues in \emph{all} of the games with probability $q^*/\ell$.
As in the proof of Theorem~\ref{thm:security-two}, this adversary ensures $q^{\mathcal{A}}_{\mathrm{v}} \geq q^*/\ell$ for any compliant slashing function.
\end{remark}

Finally, we characterize the failure probability for the optimal contract.
Suppose the contract of Section~\ref{sec:contract} is $(p_0,1-q^*_{p_0,\nu})$-tolerant per Definition~\ref{def:tolerant}, where the failure probability $q^*_{p_0,\nu}$ is depends on the total bribe $p_0$ and the nodes' utility function $U(x)=x^\nu$, \eg, $q^*_{p_0,1} = q^*$ by Theorem~\ref{thm:security-two}. 
Although Theorem~\ref{thm:security-two-risk-averse} proves that the contract of Section~\ref{sec:contract} is security optimal, unlike Theorem~\ref{thm:security-two}, its proof does not provide an explicit expression for $q^*_{p_0,\nu}$ when $\nu<1$, \ie, for risk-averse nodes.
Instead, in Appendix~\ref{sec:appendix-risk-averse-bounds}, we identify an optimization problem whose solution gives $q^*_{p_0,\nu}$.
As the optimization problem is not convex for $\nu < 1$,
in lieu of solving the problem, we provide bounds on $q^*_{p_0,\nu}$ that characterize its asymptotic behavior in terms of $\nu$, $p_0$ and $N$.
\subsection{Analysis of the Dynamic Game}
\label{sec:analysis-game}

In this section, we analyze the interaction among a rational client and the nodes during the dynamic game.
For a specified slashing function, let $q(p_0,\nu)$ denote the maximum probability that $4$-security is violated in the Nash equilibrium with the largest probability of failure, across all $p_0$-adversaries.

Theorem~\ref{thm:security-one} shows that when $k>1$, the slashing function of Section~\ref{sec:contract} achieves the minimum $q(p_0,\nu)$
among all compliant slashing functions, and this probability equals $q^*_{p_0,\nu}$.
Proofs of the subsequent theorems are presented in Appendix~\ref{sec:appendix-proofs-dynamic-game}.

\begin{theorem}
\label{thm:security-one}
Consider $(p_0,p_1)$-adversaries such that $p_1 < p_{\mathrm{comp}}-p_c$ and $p_0 < (N-f-k+1)(p_s-p_w)$.
Then, for the slashing function of Section~\ref{sec:contract}, it holds that $q(p_0,\nu) \leq q^*_{p_0,\nu}$.

Moreover, given any compliant slashing function, if $k>1$, then, there exists a $(p_0,0)$-adversary and a subgame perfect equilibrium such that $4$-security is violated in the equilibrium with probability $q^*_{p_0,\nu}$.
\end{theorem}
Theorem~\ref{thm:security-one} proves that even if the adversary does not offer any bribe to the client, \ie, $p_1=0$, if $k>1$, there exists a subgame perfect equilibrium where security is violated with the maximum probability $q^*_{p_0,\nu}$.

\begin{remark}
\label{rem:small-p1}
If $p_1 > p_{\mathrm{comp}}-p_c$, there exists a Nash equilibrium, where $4$-security is violated with overwhelming probability. 
Suppose the adversary $\mathcal{A}$ asks the nodes to \emph{not} send their clues to $\client$ or to the contract, and requests $\client$ to \emph{not} post its query to the contract.
If $\client$ never sends its query to the contract, nodes achieve a strictly better utility by accepting the adversary's offer.
Similarly, $\client$ cannot increase its utility by deviating from the adversarial action.
This is because, if $\client$ rejects its bribe and sends a query to the contract, given the nodes' actions, its payoff becomes $p_{\mathrm{comp}}-p_c$, less than the bribe $p_1$.
Hence, given $\mathcal{A}$, the specified actions indeed constitute a Nash equilibrium.
\end{remark}

Theorem~\ref{thm:security-three} analyzes the game when $k=1$.
\begin{theorem}
\label{thm:security-three}
Consider any compliant slashing function and $(p_0,p_1)$-adversaries such that $p_1 < p_{\mathrm{comp}}-p_c$ and $p_0 < (N-f-k+1)(p_s-p_w)$.
Suppose there are $N$ nodes, and $k=1 \leq N-f$.
Then, if $p_1$ satisfies
\begin{IEEEeqnarray}{C}
\label{eq:p1-constraint}
(1-q^*_{p_0,\nu})(p_f-p_c+p_1)^\nu + q^*_{p_0,\nu}(p_f-p_c+p_1+p_{\mathrm{comp}})^\nu \geq (p_f)^\nu,
\end{IEEEeqnarray}
there exists a $(p_0,p_1)$-adversary and a subgame perfect equilibrium such that $4$-security is violated with probability at least $q^*_{p_0,\nu}$, \ie, $q(p_0,\nu) \geq q^*_{p_0,\nu}.$
\end{theorem}
Via Theorems~\ref{thm:security-one} and~\ref{thm:security-three}, for all values of $k$ and all $(p_0,p_1)$-adversaries with a sufficiently large $p_1$, the slashing function of Section~\ref{sec:contract} achieves the minimum possible failure probability for $4$-security among all compliant slashing functions.
If $p_1$ satisfies formula~\eqref{eq:p1-constraint}, then the adversary can incentivize $\client$ to send a query to the contract regardless of whether $\client$ received clues over the network. 
This in turn discourages the nodes from sending clues over the network, and helps sustain an equilibrium where security rests solely on the clues sent to the contract.
In this context, slashing function of Section~\ref{sec:contract} minimizes the failure probability for security, which becomes $q^*_{p_0,\nu}$.

On the other hand, if $p_1$ is too small to satisfy formula~\eqref{eq:p1-constraint}, $p_0$ is sufficiently small (but non-zero) and $k=1$, given the optimal slashing function of Section~\ref{sec:contract}, $4$-security can be satisfied, without any query sent to the contract, with probability exceeding $q^*_{p_0,\nu}$.
This prevents the adversary from making the contract the default method for retrieving the data and bloating the blockchain.

\begin{theorem}
\label{thm:security-contract-not-used}
Consider the slashing function of Section~\ref{sec:contract} and $(p_0,p_1)$-adversaries such that $p_1 < p_{\mathrm{comp}}-p_c$, $\ p_0 < (N-f)p_w$,
and $p_1$ satisfies
\begin{IEEEeqnarray*}{C}
(1-q^*_{p_0,\nu})(p_f-p_c+p_1)^\nu + q^*_{p_0,\nu}(p_f-p_c+p_1+p_{\mathrm{comp}})^\nu < (p_f)^\nu.
\end{IEEEeqnarray*}
Suppose there are $N$ nodes, and $k=1 \leq N-f$.
Then, $4$-security is satisfied with overwhelming probability in all Nash equilibria, without the client sending its query to the contract.
\end{theorem}

When $\nu=1$, \ie, for risk-neutral nodes, formula~\eqref{eq:p1-constraint} implies $p_1 \geq p_c$.
As $p_c$ can be as small as the gas cost of sending a query, for most $(p_0,p_1)$-adversaries, we expect $p_1$ to exceed $p_c$, \ie to satisfy formula~\eqref{eq:p1-constraint}.
\subsection{Repeated Games}
\label{sec:analysis-repeated-games}

Although the adversary can offer any bribe and specify any action in return, exchange of the bribe and the execution of the action do not necessarily happen atomically.
This might discourage cooperation between the nodes and the adversary as they can renege on their promises.

One way the parties can ensure atomicity is through a \emph{trusted third party}.
It can take the custody of the nodes' internal states, along with the adversary's bribe, and adjust the payoffs after the game.
Alternatively, the adversary and nodes can sustain a cooperative equilibrium over the repeated instances of the single-stage query-response game analyzed in Section~\ref{sec:analysis-risk-neutral}, with the help of a common random coin.
For repeated games to be feasible, we assume that the nodes have more coins than the collateral $p_0$, enabling them to absorb occasional losses due to slashing in return for long-term profit.
Similarly, we assume that the adversary can continue to offer bribes each new game.
Let $\delta$ denote the discount rate.
We consider the optimal contract of Section~\ref{sec:contract} in the following analysis.

Suppose $p_0 < (N-f-k+1)(p_s-p_w)$, and there is a query at the contract at slot $t$.
By Section~\ref{sec:analysis-risk-neutral}, for any $p_0$ and $\nu \in (0,1]$, there exists a $p_0$-adversary and Nash equilibrium, where less than $k$ nodes send their clues to the contract at slot $t+1$ with probability at least $q^*_{p_0,\nu}$.
The utilities of the nodes in the equilibrium are feasible and strictly individually rational as they are at least as large as the maximum utility, $(C-p_w)^\nu$, any node can get without cooperating with the adversary.
Thus, by the Nash folk theorem, we can state the following:
\vspace{-0.3cm}
\begin{theorem}
\label{thm:folk-theorem}
There exists a discount rate $\delta^*<1$ such that for all $\delta>\delta^*$, there is a subgame perfect equilibrium of the repeated game with the same expected utilities for the nodes per game as the single-stage game.
Moreover, at each game, less than $k$ clues are posted to the contract with probability at least $q^*_{p_0,\nu}$.
\end{theorem}
By the proof of Theorem~\ref{thm:security-two-risk-averse}, no adversary can guarantee less than $k$ clues to be posted to the contract with probability larger than $q^*_{p_0,\nu}$ without making the utility of a node less than $(C-p_w)^\nu$, its minimax utility. Consequently, $q^*_{p_0,\nu}$ is the maximum failure probability that can be sustained through repeated games.

To maintain the aforementioned equilibrium, the adversary and nodes can use a grim trigger strategy:
Consider the adversary of Theorem~\ref{thm:security-two}. 
The common random coin is flipped before each game, and obtains the value $0$ with probability $q^*_{p_0,\nu}$ and the value $1$ with probability $1-q^*_{p_0,\nu}$.
Before each game, the adversary offers its bribe.
Then, the coin is flipped.
If the outcome is $0$, 
the bribed nodes are asked to withhold their clues from the contract.
At this point, if any of the bribed nodes sends its clue to the contract, the adversary stops offering any future payoff to that node.
Similarly, if the adversary fails to offer a sufficient bribe before the coin is flipped, the nodes stop cooperating with the adversary.

\section{Evaluation}
\label{sec:evaluation}
\begin{figure}
    \centering
    \includegraphics[width=0.9\linewidth]{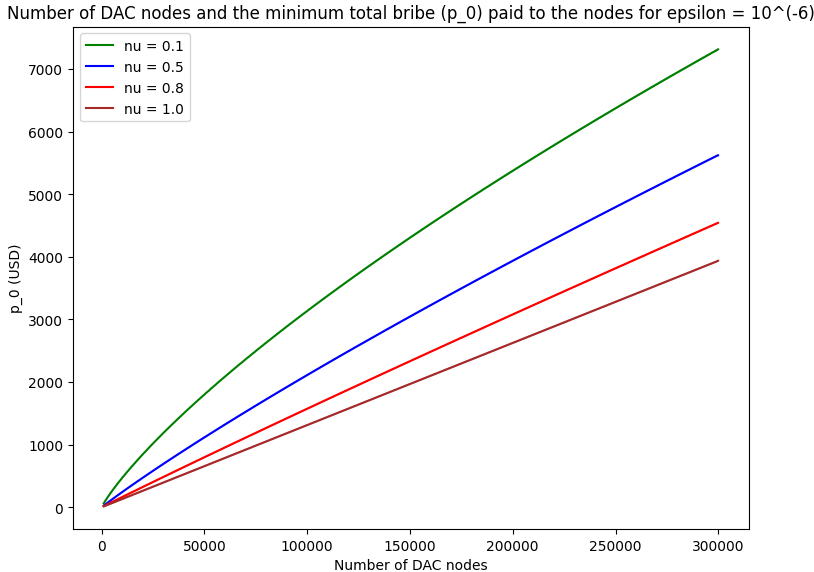}
    \caption{Lower bounds on the total bribe, $p_0$, needed to ensure that the probability the client does not obtain the answer to its query is $\epsilon = 10^{-6}$, as a function of the number of DAC nodes $N$, and the utility functions $U(x)=x^\nu$, $\nu = 0.1, 0.5, 0.8, 1.0$.}
    \label{fig:bribe-vs-n}
\end{figure}

We next calculate the bribe $p_0$ needed to violate security in the equilibrium with the largest failure probability, when a query is sent to the optimal contract of Section~\ref{sec:contract} on Ethereum.
When the clues are SNARK proofs as argued in Section~\ref{sec:model},
assuming that sending and verifying a SNARK proof on Ethereum requires $650000$ gas~\cite{zksnark-verification-cost}, and the gas cost is $34.77$ Gwei\footnote{The gas cost is the average gas price for July 15, 2022~\cite{ycharts}.}, 
we estimate the cost of posting a clue to the contract as $p_w \approx 0.0226$ ETH. 
We set the collateral $p_s$ to be $32$ ETH to match the minimum amount that can be staked in Ethereum by an independent node.
Assuming that the adversary can control up to $1/3$ of the $N$ DAC nodes, and clues from $1/3$ of the nodes are sufficient to recover the answer to the client queries, we set $N-f-k+1$ to be $N/3$.
The $1/3$ bound for the adversarial DAC nodes matches the maximum tolerable adversary fraction shown for the security of Casper FFG~\cite{casper}, the finality gadget of PoS Ethereum.
We consider $N<300,000$, which has the same magnitude as the number of validators on PoS Ethereum~\cite{val-num}.

Let $\epsilon$ denote the maximum failure probability for DAC security that the clients are willing to tolerate.
Suppose $\epsilon = 0.1\%$.
For risk-neutral nodes, Theorem~\ref{thm:security-two} implies that $\epsilon = \min(1,\frac{1}{N-f-k+1}\frac{p_0}{p_s-p_w})$.
For risk-averse nodes with the utility function $U(x)=x^\nu$, $\epsilon$ is the solution to the optimization problem in Theorem~\ref{thm:optimization-problem}, which is upper and lower by Theorem~\ref{thm:p0-bound} in Appendix~\ref{sec:appendix-risk-averse-bounds}.
Using the parameters identified above, the formula for $\epsilon$ for risk-neutral nodes and the bounds for risk-averse nodes, we calculate the following bounds\footnote{Ethereum to USD conversion rate, $1 \text{ ETH } \approx 1231.0 \text{ USD}$, is the average Ethereum price on July 15, 2022~\cite{conversion-rate}.} 
for the minimum bribe $p_0$ needed to violate security with probability $\epsilon=10^{-3}$ ($0.1\%$), as a function of the utility parameter $\nu$ (Details are presented in Appendix~\ref{sec:appendix-risk-averse-bounds}).
%\vspace{-0.5cm}
\begin{table}[h]
\begin{center}
\begin{tabular}{||c||c|c||}
 \hline
 $\nu$ & Lower bound on $p_0$ & Upper bound on $p_0$ \\
 \hhline{||=||=|=||}
 \ \ $0.1$ \ \ & \ \ $3197.9$ ETH ($3.9$ Million USD) \ \ &\ \  $13257.7$ ETH ($16.3$ Million USD) \ \ \\
 \hline
 $0.5$ & $3197.9$ ETH ($3.9$ Million USD) & $6082.5$ ETH ($7.5$ Million USD) \\
 \hline
 $0.8$ & $3197.9$ ETH ($3.9$ Million USD) & $3977.5$ ETH ($4.9$ Million USD) \\
 \hline
 $1.0$ & $3197.9$ ETH ($3.9$ Million USD) & $3197.9$ ETH ($3.9$ Million USD) \\
 \hline
\end{tabular}
\end{center}
\caption{Lower and upper bounds on %bribe 
$p_0$ in ETH and USD as a function of the utility parameter $\nu$, where $1 \text{ ETH } \approx 1231.0 \text{ USD}$, the number of DAC nodes $N$ is $300,000$, and the failure probability is $\epsilon = 0.1\%$.
}
\vspace{-0.5cm}
\label{tab:bribe-bounds}
\end{table}

The exact value of $p_0$ increases as $\nu$ decays, \ie, as the nodes become more risk averse.
This increase becomes more stark at small values of the maximum failure probability $\epsilon$.
To illustrate this point, we plot the lower bound on $p_0$ as a function of $N \in [1,300000]$, for $\nu = 0.1, 0.5, 0.8, 1.0$ and $\epsilon = 10^{-6}$ (as opposed to $\epsilon = 0.1\%$) on Figure~\ref{fig:bribe-vs-n}, where the lower bound curve increases as $\nu$ becomes smaller.
Since Table~\ref{tab:bribe-bounds} considers $\epsilon = 10^{-3}$, unlike the case with $\epsilon = 10^{-6}$, the lower bound expression for $p_0$ does not increase as $\nu$ becomes smaller.

\section{Discussion and Future Work}
\label{sec:discussion}
\paragraph{\bf Preventing centralization of storage.}
DAC members have an incentive to pool their resources
and pay for a central data repository, \eg, a cloud provider.
They then answer the client queries by querying the central repository,
and split the cost of the repository among themselves.
However, if this single repository loses the data, 
then all is lost.
Thus, a DAC protocol should discourage data centralization,
and this can be done using a cryptographic
Proofs of Replication (POR)~\cite{POR} that forces every node
to store a different incompressible version of the data. 
However, POR introduces a significant computation overhead.
Interestingly, the data centralization problem is not addressed
by data availability or storage systems such as Celestia~\cite{lazyledger} and Arweave~\cite{arweave}.

While our protocol does not solve the problem, 
arguably, it discourages data centralization.
A node that participates in a centralization scheme 
is putting its trust in the repository to preserve the data.
However, the repository has little to lose if the data is lost, 
while the node will lose its entire stake.
Hence, the node is incentivized to store the data locally
rather than to trust a third party.

\paragraph{\bf Preventing client DoS attack.}
Clients can send queries to the contract frequently, at the cost of $p_c$ coins per query.
Although $p_c$ can be as low as the gas cost of posting an account information on chain (\cf Application in Section~\ref{sec:model}), which implies a potential DoS vector, the contract can increase this cost to disincentivize DoS attacks.
The value of $p_c$ can even be adaptively chosen as a function of the number of queries to reduce congestion. 
Then, as long as $p_c$ is not subsidized by the bribe $p_1$, no rational client would send a query unless the nodes withhold their clues.
However, $p_c$ should not be too high as that would hurt the contract balance by requiring a high $p_{\mathrm{comp}}$ (\cf Remark~\ref{rem:small-p1}), and discourage rational clients from sending queries for accounts with smaller balances (\ie, $p_f$).
An interesting future work is to determine the optimal $p_c$ that would not impose a high burden on most accounts while making spamming attacks costly.

\paragraph{\bf Utility functions.}
The analysis in Section~\ref{sec:evaluation} demonstrates how risk-aversion implies a higher bribe for the adversary to violate security.
However, the exact shape of the utility function depends on the marginal utility for the coin in which the payoffs are provided.
Quantifying this marginal utility and identifying the correct function is important future work to accurately assess the affect of bribery on security.

\paragraph{\bf Bribery.}
Besides external offers, bribery quantifies the nodes' incentives to withhold data, \eg, to prevent a competitor from learning about a transaction.
In general, bribery captures all incentives (\eg, MEV) that might cause the rational nodes to deviate from the prescribed consensus and data availability protocols.

\paragraph{\bf Collusion.} 
Collusions among nodes can help mitigate the free-rider problem in Section~\ref{sec:analysis-risk-neutral}.
When the nodes are non-colluding and $p_w = 0$, there exists an equilibrium with non-negligible probability of failure, where the nodes trust each other to send sufficiently many clues to the contract.
By colluding, rational nodes can post exactly $k$ clues after each query, thus obviating the punishment $p_w$ for non-responsive nodes when security is satisfied.
This minimizes the number of clues posted to the contract.
Analyzing collusions among subsets of nodes can also shed light on games, where multiple nodes are controlled by the same entity.

\iftoggle{fullversion}{
\paragraph{\bf Acknowledgments.}
This work was supported by NSF, ONR, the Simons Foundation, NTT Research, and a grant from Ripple.
Additional support was provided by the Stanford Center for Blockchain Research.
}{}

\bibliographystyle{splncs04}
\bibliography{references}

\begin{thebibliography}{10}
\providecommand{\url}[1]{\texttt{#1}}
\providecommand{\urlprefix}{URL }
\providecommand{\doi}[1]{https://doi.org/#1}

\bibitem{EigenLayer}
Eigenlayer, \url{https://www.layrlabs.com/}, {Accessed: 2022-08-01}

\bibitem{ycharts}
Ethereum average gas price.
  \url{https://ycharts.com/indicators/ethereum_average_gas_price}, {Accessed:
  2022-07-15}

\bibitem{conversion-rate}
Ethereum historical data.
  \url{https://www.investing.com/crypto/ethereum/historical-data}, {Accessed:
  2022-07-15}

\bibitem{starkex}
{StarkEx v4}, \url{https://docs.starkware.co/starkex-v4/}, {Accessed:
  2022-03-01}

\bibitem{zkporter}
{zkPorter: a breakthrough in L2 scaling} (2021),
  \url{https://blog.matter-labs.io/zkporter-a-breakthrough-in-l2-scaling-ed5e48842fbf}

\bibitem{lazyledger}
Al-Bassam, M.: Lazyledger: A distributed data availability ledger with
  client-side smart contracts. arXiv:1905.09274  (2019),
  \url{https://arxiv.org/abs/1905.09274}

\bibitem{albassam2018fraud}
Al{-}Bassam, M., Sonnino, A., Buterin, V., Khoffi, I.: Fraud and data
  availability proofs: Detecting invalid blocks in light clients. In: Financial
  Cryptography {(2)}. Lecture Notes in Computer Science, vol. 12675, pp.
  279--298. Springer (2021)

\bibitem{val-num}
Anderrson, S.: {ETH 2.0 crosses 300,000 validators, Ether deposits worth 28.9B
  already locked} (2022),
  \url{https://www.thecoinrepublic.com/2022/03/05/eth-2-0-crosses-300000-validators-ether-deposits-worth-28-9b-already-locked/}

\bibitem{chainlink-2}
Breidenbach, L., Cachin, C., Chan, B., Coventry, A., Ellis, S., Juels, A.,
  Koushanfar, F., Miller, A., Magauran, B., Moroz, D., Nazarov, S., Topliceanu,
  A., Tramer, F., Zhang, F.: Chainlink 2.0: Next steps in the evolution of
  decentralized oracle networks. Whitepaper (2021),
  \url{https://research.chain.link/whitepaper-v2.pdf}

\bibitem{zksnark-verification-cost}
Buterin, V.: On-chain scaling to potentially $\sim$500 tx/sec through mass tx
  validation (2018),
  \url{https://ethresear.ch/t/on-chain-scaling-to-potentially-500-tx-sec-through-mass-tx-validation/3477}

\bibitem{KZGdas}
Buterin, V.: 2d data availability with kate commitments (2020),
  \url{https://ethresear.ch/t/2d-data-availability-with-kate-commitments/8081}

\bibitem{protodanksharding}
Buterin, V.: Proto-danksharding faq (2022),
  \url{https://notes.ethereum.org/@vbuterin/proto_danksharding_faq#If-data-is-deleted-after-30-days-how-would-users-access-older-blobs}

\bibitem{casper}
Buterin, V., Griffith, V.: {Casper} the friendly finality gadget.
  arXiv:1710.09437  (2019), \url{https://arxiv.org/abs/1710.09437}

\bibitem{ransom}
Drake, J.: Starkex validium ransom attack (2020),
  \url{https://notes.ethereum.org/DD7GyItYQ02d0ax_X-UbWg?view}

\bibitem{chainlink-1}
Ellis, S., Juels, A., Nazarov, S.: Chainlink a decentralized oracle network.
  Whitepaper (2017), \url{https://research.chain.link/whitepaper-v1.pdf}

\bibitem{proof-of-custody}
Feist, D.: Proofs of custody (2021),
  \url{https://dankradfeist.de/ethereum/2021/09/30/proofs-of-custody.html}

\bibitem{danksharding}
Feist, D.: New sharding design with tight beacon and shard block integration
  (2022), \url{https://notes.ethereum.org/@dankrad/new_sharding}

\bibitem{POR}
Fisch, B.: Poreps: Proofs of space on useful data. Cryptology ePrint
  Archive:2018/678  (2018), \url{https://eprint.iacr.org/2018/678}

\bibitem{survey-2}
Gudgeon, L., Moreno{-}Sanchez, P., Roos, S., McCorry, P., Gervais, A.: Sok:
  Layer-two blockchain protocols. In: Financial Cryptography. Lecture Notes in
  Computer Science, vol. 12059, pp. 201--226. Springer (2020)

\bibitem{kate}
Kate, A., Zaverucha, G.M., Goldberg, I.: Constant-size commitments to
  polynomials and their applications. In: {ASIACRYPT}. Lecture Notes in
  Computer Science, vol.~6477, pp. 177--194. Springer (2010)

\bibitem{rollup}
McCorry, P., Buckland, C., Yee, B., Song, D.: Sok: Validating bridges as a
  scaling solution for blockchains. Cryptology ePrint Archive:2021/1589
  (2021), \url{https://eprint.iacr.org/2021/1589}

\bibitem{permacoin}
Miller, A., Juels, A., Shi, E., Parno, B., Katz, J.: Permacoin: Repurposing
  bitcoin work for data preservation. In: {IEEE} Symposium on Security and
  Privacy. pp. 475--490. {IEEE} Computer Society (2014)

\bibitem{validium}
Nazirkhanova, K., Neu, J., Tse, D.: Information dispersal with provable
  retrievability for rollups. arXiv:2111.12323  (2021),
  \url{https://arxiv.org/abs/2111.12323}, in ACM Advances in Financial
  Technologies - AFT 2022

\bibitem{avail}
Polygon: Avail {-} the data availability blockchain (2021),
  \url{https://github.com/maticnetwork/data-availability}

\bibitem{filecoin}
Psaras, Y., Dias, D.: The interplanetary file system and the filecoin network.
  In: {DSN} (Supplements). p.~80. {IEEE} (2020)

\bibitem{da-offchain}
Sriram, A., Adler, J.: {The Ethereum Off-Chain Data Availability Landscape}
  (2022),
  \url{https://blog.celestia.org/ethereum-off-chain-data-availability-landscape/}

\bibitem{celestium}
Sriram, A., Adler, J., Al-Bassam, M.: Quantum gravity bridge: Secure off-chain
  data availability for ethereum l2s with celestia (2022),
  \url{https://blog.celestia.org/celestiums/}

\bibitem{arweave}
Williams, S., Diordiiev, V., Berman, L., Raybould, I., Uemlianin, I.: Arweave:
  A protocol for economically sustainable information permanence. Yellow paper
  (2019), \url{https://www.arweave.org/yellow-paper.pdf}

\bibitem{survey-1}
Yee, B., Song, D., McCorry, P., Buckland, C.: Shades of finality and layer 2
  scaling. arXiv:2201.07920  (2022), \url{https://arxiv.org/abs/2201.07920}

\bibitem{codedmerkletree}
Yu, M., Sahraei, S., Li, S., Avestimehr, S., Kannan, S., Viswanath, P.: Coded
  merkle tree: Solving data availability attacks in blockchains. In: Financial
  Cryptography. Lecture Notes in Computer Science, vol. 12059, pp. 114--134.
  Springer (2020)

\end{thebibliography}

\appendix
\section{Proofs of Lemma~\ref{lem:security-zero}, and Theorems~\ref{thm:axioms-converse} and~\ref{thm:security-two} for risk-neutral nodes}
\label{sec:appendix-proofs-risk-neutral}

\begin{proof}[Proof of Lemma~\ref{lem:security-zero} for risk-neutral nodes and clients]
Consider a game played among the client and the DAC nodes.
At the beginning of slot $2$ the client $\client$ either received $k$ or more valid clues over the network from the nodes, or it did not.
\begin{itemize}[nosep]
\item 
Let $Q_{\geq k}$ denote the event that $\client$ takes the action $\query$, \ie, sends query to the contract, by slot $2$ even if it receives $k$ or more valid clues over the network by the end of slot $3$.

\item 
Let $Q_{< k}$ denote the event that $\client$ takes the action $\query$ by slot $2$ if it does not receive $k$ or more valid clues over the network by the end of slot $3$.

\item Let $R$ denote the event that $k$ or more nodes take the action $\reply$, \ie send clues to the client over the network, by slot $3$.

\item Let $P$ denote the event $k$ or more nodes take the action $\place$, \ie, post clues to the contract, at slot $t+1 \leq 3$ if a query is received by the contract by the end of some slot $t \leq 2$.
\end{itemize}

\medskip\noindent
Given the events above, we can summarize the $\client$'s payoff by the following table:
\begin{table}[h]
\begin{center}
\begin{tabular}{x{2cm}|x{3cm}|x{3cm}|x{2cm}|x{2cm}} 
 & $Q_{\geq k} \land Q_{< k}$ & $Q_{\geq k} \land \overline{Q}_{< k}$ & $\overline{Q}_{\geq k} \land Q_{< k}$ & $\overline{Q}_{\geq k} \land \overline{Q}_{< k}$ \\
 \hline
 \hline
 $R \land P$ & $p_f-p_c$ & $p_f-p_c$ & $p_f$ & $p_f$ \\
 \hline
 $\overline{R} \land P$ & $p_f-p_c$ & $0$ & $p_f-p_c$ & $0$ \\
 \hline
 $R \land \overline{P}$ & $p_f+p_{\mathrm{comp}}-p_c$ & $p_f+p_{\mathrm{comp}}-p_c$ & $p_f$ & $p_f$ \\
 \hline
 $\overline{R} \land \overline{P}$ & $p_{\mathrm{comp}}-p_c$ & $0$ & $p_{\mathrm{comp}}-p_c$ & $0$
\end{tabular}
\end{center}
\caption{Payoff of a risk-neutral client}
\end{table}

Towards contradiction, suppose there exists a Nash equilibrium, where the probability $\Pr[\overline{R} \land \overline{P}]$ is non-zero.
If $\Pr[\overline{R} \land \overline{P}]>0$, then the client never takes the actions $Q_{\geq k} \land \overline{Q}_{< k}$ and $\overline{Q}_{\geq k} \land \overline{Q}_{< k}$ with positive probability in the equilibrium, as it can increase its expected payoff by reducing the probability of $Q_{\geq k} \land \overline{Q}_{< k}$ in favor of $Q_{\geq k} \land Q_{< k}$, and by reducing the probability of $\overline{Q}_{\geq k} \land \overline{Q}_{< k}$ in favor of $\overline{Q}_{\geq k} \land Q_{< k}$.
Hence, in the equilibrium, either $\client$ receives $k$ or more valid clues over the network by the end of slot $3$, \ie, the event $R$ happens, or $\client$ sends a query to the contract by slot $2$.
In the latter case, $\client$'s query is received by the contract by the end of slot $2$.

If a valid query appears in the contract by the end of some slot $t$, and the node $\node_i$, $i \in [N-f]$, sends its clue to the contract at slot $t+1$, its payoff becomes $-p_w$.
On the other hand, if a valid query is received by the contract by the end of some slot $t$, and $\node_i$ does not send its clue to the contract at slot $t+1$, its payoff can at most be $-p_w-\epsilon$.
Since $-p_w > -p_w-\epsilon$, in the equilibrium, if $\client$'s query is received by the contract by the end of slot $2$, $\node_i$ sends its clue to the contract by slot $3$, which is then observed by the client by the beginning of slot $4$.
As this holds for all nodes $\node_i$, $i \in [N-f]$, if a query is received by the contract by the end of slot $2$, all non-adversarial nodes send their clues to the contract by slot $3$, \ie, the event $P$ happens.
However, this implies $\Pr[R \lor P]=1$, and
$\Pr[\overline{R} \land \overline{P}]=0$, which is a contradiction.
Consequently, $4$-security is satisfied with overwhelming probability in all Nash equilibria.
\end{proof}

\begin{proof}[Proof of Theorem~\ref{thm:axioms-converse} for risk-neutral nodes]
Suppose a query is received by the contract at some slot $t \in \mathbb{N}$.
Consider the adversary $\mathcal{A}$ that makes every adversarial node take the action $\neg\place$ (not post clues to the contract) at all slots.
Suppose there exists a Nash equilibrium of this subgame, where each node $\node_i$, $i \in [N-f]$, independently decides to take the action $\place$ at slot $t+1$, with probability $r^* \in [0,1)$, and to take the action $\neg\place$ at slot $t+1$, with probability $1-r^*$.
Then, in the equilibrium, the expected payoff of each node $\node_i$, $i \in [N-f]$, becomes 
\begin{IEEEeqnarray*}{C}
r^* (-p_w + \mathbb{E}[f_i(\mathbf{X})|X_i = 1]) + (1-r^*)\mathbb{E}[f_i(\mathbf{X})|X_i = 0]
\end{IEEEeqnarray*}
Here, $\mathbb{E}[f_i(\mathbf{X})|X_i = 1]$ and $\mathbb{E}[f_i(\mathbf{X})|X_i = 0]$ denote the expected payoff of $\node_i$ given all other nodes' actions and conditioned on the fact that $\node_i$ takes the action $\place$ and $\neg\place$ at slot $t+1$ respectively.
As the slashing function is symmetric and offers no rewards, there exist $e_0,e_1 \in [-p_s,0]$ such that $\mathbb{E}[f_i(\mathbf{X})|X_i = 1]=\mathbb{E}[f_j(\mathbf{X})|X_j = 1]=e_1$, and $\mathbb{E}[f_i(\mathbf{X})|X_i = 0]=\mathbb{E}[f_j(\mathbf{X})|X_j = 0]=e_0$ for all $i,j \in [N-f]$. 

For the action profile described above to be a Nash equilibrium, it must be the case that 
\begin{IEEEeqnarray*}{C}
r^* (-p_w + e_1) + (1-r^*)e_0 \geq r (-p_w + e_1) + (1-r) e_0
\end{IEEEeqnarray*} 
for all $r \in [0,1]$.
For each $i \in [N-f]$, the inequality above is satisfied by
\begin{itemize}
    \item $r^*=1$, if $-p_w + e_1 > e_0$,
    \item $r^*=0$, if $-p_w + e_1 < e_0$,
    \item any $r^* \in [0,1]$, if $-p_w + e_1 = e_0$.
\end{itemize}

Let $r^*_{-i} \in [0,1]$ denote the probability that among the $N-f-1$ non-adversarial nodes other than $\node_i$, at most $k-1$ nodes take the action $\place$ (send clues to the contract) at slot $t+1$ in the equilibrium.
As the slashing function is symmetric, $r^*_{-i}=r^*_{-j}=r^*_{-1}$ for all $i,j \in [N-f]$, and
\begin{IEEEeqnarray*}{rCl}
- (1-r^*_{-1})B - r^*_{-1}p_s &\leq& e_0 \leq 0 \\
-p_w-p_s &\leq& -p_w + e_1 \leq -p_w,
\end{IEEEeqnarray*}
where $B < p_w$.
This implies a value $r^*_{-1} \in (0,1]$ such that $-p_w+e_1 \leq e_0$, and $r^*<1$.

If $-p_w+e_1 < e_0$, then $r^*=0$.
On the other hand, if $-p_w+e_1 = e_0$, as $B<p_w$, it must be the case that $r^*_{-1} \in (0,1]$, which implies $r^* \in [0,1)$.
Consequently, if $B < p_w$, there indeed exists a $(0,0)$-adversary $\mathcal{A}$ and a Nash equilibrium, where each non-adversarial node independently decides to not post its clue to the contract at the slot after a query is received by the contract, with probability $r^*>0$.

Finally, suppose there exists a subgame perfect equilibrium of the dynamic game, where none of the nodes send a valid clue to the client $\client$ over the network by slot $3$, and $\client$ sends its query to the contract by slot $2$.
Once the query appears in the contract, each nodes chooses to withhold its clue from the contract in the next slot with probability $r^*$.

As shown above, if there is a query in the contract, none of the nodes can increase its expected payoff by deviating from the specified action given the other nodes' actions.
Similarly, $\client$ cannot increase its expected payoff by not sending a query to the contract, when none of the nodes sends a valid clue over the network by slot $3$.
Finally, none of the nodes can increase its expected payoff by sending a valid clue to $\client$ over the network by slot $2$, since this does not affect $\client$'s behavior.
This is because $k>1$ and all other nodes refuse to send their clues over the network.
Hence, the claimed action profile indeed constitutes a subgame perfect equilibrium.
However, in this case, there are less than $k$ valid clues in the contract by slot $4$ with probability at least $(1-r^*)^{N-f}>0$.
Thus, there exists a $(0,0)$-adversary and a subgame perfect equilibrium, where $4$-security is violated with positive probability.
\end{proof}

\begin{proof}[Proof of Theorem~\ref{thm:security-two}]
We first show that given the slashing function of Section~\ref{sec:contract}, $q^{\mathcal{A}}_{\mathrm{v}} \leq q^*$ for any $p_0$-adversary $\mathcal{A}$.
Let $F$ denote the event that there are less than $k$ valid clues in the contract at slot $t+1$.
Towards contradiction, suppose there exists a Nash equilibrium such that $\Pr[F] > q^*$ in the equilibrium.
Then, with probability greater than $q^*$, at least $N-f-k+1$ non-adversarial nodes take the action $\neg \place$, \ie, do not post their clues to the contract, at slot $t+1$ in the equilibrium.
Let $U^*_i$ denote the realization of the node $\node_i$'s payoff, $i \in [N-f]$, in the equilibrium. 
In the case of event $F$, at least $N-f-k+1$ nodes incur a penalty of $-p_s$, and given the total bribe $p_0$, the total expected payoff of the non-adversarial nodes can at most be 
\begin{IEEEeqnarray*}{C} 
\sum_{i=1}^{N-f}\mathbb{E}[U^*_i|F] \leq p_0-(k-1)p_w-(N-f-k+1)p_s
\end{IEEEeqnarray*}
in the equilibrium.
Conversely, if the event $F$ does not happen, the total expected payoff can at most be
\begin{IEEEeqnarray*}{C} 
\sum_{i=1}^{N-f} \mathbb{E}[U^*_i|\overline{F}] \leq p_0-(N-f)p_w
\end{IEEEeqnarray*}
in the equilibrium.
Using the two inequalities and the assumed lower bound on $\Pr[F]$, we derive the following upper bound on the total expected payoff of the non-adversarial nodes in the equilibrium:
\begin{IEEEeqnarray*}{rCl} 
\sum_{i=1}^{N-f} \mathbb{E}[U^*_i] &\leq& p_0-(N-f)p_w-\Pr[F](N-f-k+1)(p_s-p_w) \\
&<& p_0-(N-f)p_w-q^*(N-f-k+1)(p_s-p_w) \\
&=& -(N-f)p_w  
\end{IEEEeqnarray*}
The inequality above implies the existence of at least one node $\node_{i^*}$, $i^* \in [N-f]$, such that $\mathbb{E}[U^*_{i^*}] < -p_w$.
However, sending a valid clue to the contract at slot $t+1$ gives a payoff of $-p_w$ for $\node_{i^*}$ regardless of the actions of all other nodes, implying that there exists an action that strictly dominates the one taken by $\node_{i^*}$ in the equilibrium.
Thus, we have reached a contradiction.

\begin{center}
    ***
\end{center}

We next construct a $p_0$-adversary $\mathcal{A}$ such that for any compliant slashing function, there exists a Nash equilibrium, where $q^{\mathcal{A}}_{\mathrm{v}} \geq q^*$.
Let $X_i$ denote the indicator random variable for the event that $\node_i$ takes the action $\place$, \ie sends its query to the contract, at slot $t+1$.
Then, any adversary $\mathcal{A}$ can be characterized as follows:
\begin{itemize}
    \item Bribes offered to the non-adversarial nodes $\node_i$, $i \in [N-f]$: $p^i_b$, $i \in [N-f]$.
    Bribes must satisfy the equation $\sum_{i \in [N-f]} p^i_b \leq p_0$.
    A non-adversarial node is called \emph{corrupt} if it accepts the bribe.
    \item The probability distribution over the actions adopted by the adversarial and corrupt nodes.
    For each $Q$, $\{N-f+1,\ldots,N\} \subseteq Q \subseteq [N]$, representing the set of adversarial and corrupt nodes, $\mathcal{A}$ proposes the following action profile:
    $\Pr[X_i=x_i \in \{0,1\}, i \in Q] = q_{(x_i, i \in Q)}$.
\end{itemize}
Similarly, given any adversary $\mathcal{A}$, each Nash equilibrium can be described by the following variables: $\{\tilde{r}^*_i\}_{i \in [N-f]}$ and $\{r^*_i\}_{i \in [N-f]}$.
Here, $\tilde{r}^*_i$ denotes the probability that $\node_i$ accepts the bribe in the equilibrium, whereas $r^*_i$ denotes the probability that $\node_i$ takes the action $\place$ in the event that it does not accept the bribe.
We allow $r^*_i$ to be undefined (shown as $-$) if $\tilde{r}^*_i=1$, \ie, if $\node_i$ accepts the adversary's bribe, in which case it will take the action $\place$ as dictated by the adversary.
In the equilibrium, each node $\node_j$, $j \in [N-f]$, chooses the values $\tilde{r}^*_j$ and $r^*_j$ to maximize its expected payoff given the adversary $\mathcal{A}$, and the other nodes' equilibrium actions, \ie, $\{\tilde{r}^*_i\}_{i \in [N-f]/\{j\}}$ and $\{r^*_i\}_{i \in [N-f]/\{j\}}$.

Consider the $p_0$-adversary $\mathcal{A}$ that offers a bribe of $p^i_b = p_0/(N-f-k+1)$ to the nodes $\node_i$, $i \in [N-f-k+1]$, and for the set $Q \subseteq [N-f-k+1] \cup \{N-f+1,\ldots,N\}$, specifies
\begin{IEEEeqnarray*}{rCl} 
\Pr[X_i = 0, i \in Q] &=& \frac{p_0}{(N-f-k+1)(p_s-p_w)}, \\
\Pr[X_i=1, i \in Q] &=& 1-\frac{p_0}{(N-f-k+1)(p_s-p_w)}.
\end{IEEEeqnarray*}
We will show that given $\mathcal{A}$, for any compliant slashing function, there exists a Nash equilibrium such that $q^{\mathcal{A}}_{\mathrm{v}} \geq q^*$.
In the equilibrium, for each $i \in [N-f-k+1]$, either $(\tilde{r}^*_i,r^*_i) = (1,-)$ or $(\tilde{r}^*_i,r^*_i) = (0,0)$, and for each $i \in [N-f]/[N-f-k+1]$, $(\tilde{r}^*_i,r^*_i) = (0,0)$. 
In other words, nodes that are offered bribes either accept the bribe and become corrupted, or do not post their clues to the contract in the equilibrium, whereas those that are not offered bribes post their clues to the contract.

Since the considered slashing functions are compliant, they offer no rewards and satisfy $B$ minimal punishment for some $B \geq p_w$ by Theorem~\ref{thm:axioms-converse}.
Thus, for all $i \in [N]$, it must be the case that $f_i(\mathbf{x}) \leq 0$ for all $\mathbf{x} \in \{0,1\}^N$, and $f_i(\mathbf{x}) \leq -p_w$ for all $\mathbf{x} \in \{0,1\}^N$ such that $x_i=0$.
Then, if a node $\node_i$, $i \in [N-f-k+1]$, rejects the bribe, its expected payoff becomes
\begin{IEEEeqnarray*}{rCl} 
(-p_w+\mathbb{E}[f_i(\mathbf{X})|X_i = 1])r^*_i + \mathbb{E}[f_i(\mathbf{X})|X_i = 0](1-r^*_i) \leq \mathbb{E}[f_i(\mathbf{X})|X_i = 1]r^*_i-p_w.
\end{IEEEeqnarray*}
Here, $\mathbb{E}[f_i(\mathbf{X})|X_i = 1]$ and $\mathbb{E}[f_i(\mathbf{X})|X_i = 0]$ denote the expected payoff of $\node_i$ given all other nodes' actions in the claimed equilibrium and conditioned on the fact that $\node_i$ takes the actions $\place$ and $\neg\place$ at slot $t+1$ respectively.
As the slashing function is symmetric, for any $i,j \in [N-f-k+1]$, $\mathbb{E}[f_i(\mathbf{X})|X_i = 1]=\mathbb{E}[f_j(\mathbf{X})|X_j = 1]=e_1 \leq 0$ and $\mathbb{E}[f_i(\mathbf{X})|X_i = 0]=\mathbb{E}[f_j(\mathbf{X})|X_j = 0]=e_0 \leq -B \leq -p_w$.

As each node $\node_i$, $i \in [N-f-k+1]$, maximizes its payoff given all other nodes' payoffs in the equilibrium, if $\node_i$ rejects the bribe, it must be the case that
\begin{itemize}
    \item $r^*_i=1$, if $-p_w+e_1 > e_0$. In this case, $\node_i$'s expected payoff becomes $-p_w+e_1 \leq -p_w$.
    \item $r^*_i=0$, if $-p_w+e_1 < e_0$.
    \item $r^*_i$ can be any value in $[0,1]$, if $-p_w+e_1 = e_0$.
    In this case, $\node_i$'s expected payoff becomes $e_0 \leq -B \leq -p_w$
\end{itemize}
On the other hand, if $\node_i$ accepts the bribe, its expected payoff becomes at least
\begin{IEEEeqnarray*}{rCl} 
&& \frac{p_0}{N-f-k+1} - p_s\frac{p_0}{(N-f-k+1)(p_s-p_w)} \\
&+& (-p_w + \mathbb{E}[f_i(\mathbf{X})|X_i = 1])\left(1-\frac{p_0}{(N-f-k+1)(p_s-p_w)}\right) \\
&=& -p_w + e_1\left(1-\frac{p_0}{(N-f-k+1)(p_s-p_w)}\right).
\end{IEEEeqnarray*}

Since $e_1 \leq 0$, if $-p_w+e_1 \geq e_0$, then the expected payoff of node $\node_i$ when it accepts the bribe is at least as large as its expected payoff when it rejects the bribe.
Hence, if $-p_w+e_1 \geq e_0$, for the nodes $\node_i$, $i \in [N-f-k+1]$, there does not exist any action that strictly dominates $(\tilde{r}^*_i,r^*_i)=(1,-)$.
In this case, the action profile specified by the adversary constitutes a Nash equilibrium, and $q^{\mathcal{A}}_{\mathrm{v}}$ becomes at least
\begin{IEEEeqnarray*}{C}
\Pr[X_i = 0, i \in Q] = \frac{p_0}{(N-f-k+1)(p_s-p_w)} = q^*,
\end{IEEEeqnarray*}
where $Q=[N-f-k+1]\cup\{N-f+1,\ldots,N\}$.
Conversely, if $-p_w+e_1 < e_0$, then either the nodes $\node_i$, $i \in [N-f-k+1]$, all reject the bribe, and set $r^*_i=0$, which implies $q^{\mathcal{A}}_{\mathrm{v}} = 1$, or they all accept the bribe, which implies $q^{\mathcal{A}}_{\mathrm{v}} \geq q^*$.
In both cases, $q^{\mathcal{A}}_{\mathrm{v}} \geq q^*$, thus concluding the proof.
\end{proof}
\section{Proofs of Lemma~\ref{lem:security-zero}, and Theorems~\ref{thm:axioms-converse}, and~\ref{thm:security-two-risk-averse} for risk-averse nodes}
\label{sec:appendix-proofs-risk-averse}

\begin{proof}[Proof of Lemma~\ref{lem:security-zero} for risk-averse nodes and clients]
Consider a game played among the client and the DAC nodes.
At the beginning of slot $2$, the client $\client$ either received $k$ or more valid clues over the network from the nodes, or it did not.
Recall the definitions of the events $Q_{\geq k}$, $Q_{< k}$, $R$ and $P$ from the proof of Lemma~\ref{lem:security-zero} for risk-neutral nodes in Appendix~\ref{sec:appendix-proofs-risk-neutral}.
Given these events, we can summarize the $\client$'s payoff by the following table:
\begin{table}[h]
\begin{center}
\begin{tabular}{x{2cm}|x{3cm}|x{3cm}|x{2cm}|x{2cm}} 
 & $Q_{\geq k} \land Q_{< k}$ & $Q_{\geq k} \land \overline{Q}_{< k}$ & $\overline{Q}_{\geq k} \land Q_{< k}$ & $\overline{Q}_{\geq k} \land \overline{Q}_{< k}$ \\
 \hline
 \hline
 $R \land P$ & $(p_f-p_c)^\nu$ & $(p_f-p_c)^\nu$ & $(p_f)^\nu$ & $(p_f)^\nu$ \\
 \hline
 $\overline{R} \land P$ & $(p_f-p_c)^\nu$ & $0$ & $(p_f-p_c)^\nu$ & $0$ \\
 \hline
 $R \land \overline{P}$ & $(p_f+p_{\mathrm{comp}}-p_c)^\nu$ & $(p_f+p_{\mathrm{comp}}-p_c)^\nu$ & $(p_f)^\nu$ & $(p_f)^\nu$ \\
 \hline
 $\overline{R} \land \overline{P}$ & $(p_{\mathrm{comp}}-p_c)^\nu$ & $0$ & $(p_{\mathrm{comp}}-p_c)^\nu$ & $0$
\end{tabular}
\end{center}
\caption{Utility of a risk-averse client}
\label{tab:client-utility}
\end{table}

Towards contradiction, suppose there exists a Nash equilibrium, where the probability $\Pr[\overline{R} \land \overline{P}]$ is non-zero.
If $\Pr[\overline{R} \land \overline{P}]>0$, then the client never takes the actions $Q_{\geq k} \land \overline{Q}_{< k}$ and $\overline{Q}_{\geq k} \land \overline{Q}_{< k}$ with positive probability in the equilibrium, as it can increase its expected payoff by reducing the probability of $Q_{\geq k} \land \overline{Q}_{< k}$ in favor of $Q_{\geq k} \land Q_{< k}$, and by reducing the probability of $\overline{Q}_{\geq k} \land \overline{Q}_{< k}$ in favor of $\overline{Q}_{\geq k} \land Q_{< k}$.
Hence, in the equilibrium, either $\client$ receives $k$ or more valid clues over the network by the end of slot $3$, \ie, the event $R$ happens, or $\client$ sends a query to the contract by slot $2$.
In the latter case, $\client$'s query is received by the contract by the end of slot $2$.

If a valid query appears in the contract by the end of some slot $t$, and the node $\node_i$, $i \in [N-f]$, sends its clue to the contract at slot $t+1$, its payoff becomes $(C-p_w)^\nu$.
On the other hand, if a valid query is received by the contract by the end of some slot $t$, and $\node_i$ does not send its clue to the contract at slot $t+1$, its payoff can at most be $(C-p_w-\epsilon)^\nu$.
Since $(C-p_w)^\nu > (C-p_w-\epsilon)^\nu$, in the equilibrium, if $\client$'s query is received by the contract by the end of slot $2$, $\node_i$ sends its clue to the contract by slot $3$, which is then observed by the client by the beginning of slot $4$.
As this holds for all nodes $\node_i$, $i \in [N-f]$, if a query is received by the contract by the end of slot $2$, all non-adversarial nodes send their clues to the contract by slot $3$, \ie, the event $P$ happens.
However, this implies $\Pr[R \lor P]=1$, and
$\Pr[\overline{R} \land \overline{P}]=0$, which is a contradiction.
Consequently, $4$-security is satisfied with overwhelming probability in all Nash equilibria.
\end{proof}

Theorem~\ref{thm:axioms} follows from Lemma~\ref{lem:security-zero}, which shows security under no attack (A3) for the slashing function of Section~\ref{sec:contract}.

\begin{proof}[Proof of Theorem~\ref{thm:axioms-converse} for risk-averse nodes]
Suppose a query is received by the contract at some slot $t \in \mathbb{N}$.
Consider the adversary $\mathcal{A}$ that makes every adversarial node take the action $\neg\place$ (not post clues to the contract) at all slots.
Suppose there exists a Nash equilibrium of this subgame, where each node $\node_i$, $i \in [N-f]$, independently decides to take the action $\place$ at slot $t+1$, with probability $r^* \in [0,1)$, and to take the action $\neg\place$ at slot $t+1$, with probability $1-r^*$.
Then, in the equilibrium, the expected payoff of each node $\node_i$, $i \in [N-f]$, becomes 
\begin{IEEEeqnarray*}{C}
r^* \mathbb{E}[(C-p_w+f_i(\mathbf{X}))^\nu|X_i = 1] + (1-r^*)\mathbb{E}[(C+f_i(\mathbf{X}))^\nu|X_i = 0],
\end{IEEEeqnarray*}
where $C=p_s+p_w$.
Here, $\mathbb{E}[(C-p_w+f_i(\mathbf{X}))^\nu|X_i = 1]$ and $\mathbb{E}[(C+f_i(\mathbf{X}))^\nu|X_i = 0]$ denote the expected payoff of $\node_i$ given all other nodes' actions and conditioned on the fact that $\node_i$ takes the action $\place$ and $\neg\place$ at slot $t+1$ respectively.
As the slashing function is symmetric and offers no rewards, there exist $e_0,e_1 \in [-p_s,0]$ such that $\mathbb{E}[(C-p_w+f_i(\mathbf{X}))^\nu|X_i = 1]=\mathbb{E}[(C-p_w+f_j(\mathbf{X}))^\nu|X_j = 1]=e_1$, and $\mathbb{E}[(C+f_i(\mathbf{X}))^\nu|X_i = 0]=\mathbb{E}[(C+f_j(\mathbf{X}))^\nu|X_j = 0]=e_0$ for all $i,j \in [N-f]$. 

For the action profile described above to be a Nash equilibrium, it must be the case that 
\begin{IEEEeqnarray*}{C}
r^* e_1 + (1-r^*)e_0 \geq r e_1 + (1-r) e_0
\end{IEEEeqnarray*} 
for all $r \in [0,1]$.
For each $i \in [N-f]$, the inequality above is satisfied by
\begin{itemize}
    \item $r^*=1$, if $e_1 > e_0$,
    \item $r^*=0$, if $e_1 < e_0$,
    \item any $r^* \in [0,1]$, if $e_1 = e_0$.
\end{itemize}

Let $r^*_{-i} \in [0,1]$ denote the probability that among the $N-f-1$ non-adversarial nodes other than $\node_i$, at most $k-1$ nodes take the action $\place$ (send clues to the contract) at slot $t+1$ in the equilibrium.
As the slashing function is compliant, $r^*_{-i}=r^*_{-j}=r^*_{-1}$ for all $i,j \in [N-f]$, and
\begin{IEEEeqnarray*}{rCl}
(1-r^*_{-1})(C-B)^\nu - r^*_{-1}(C-p_s)^\nu &\leq& e_0 \leq (C)^\nu \\
(C-p_w-p_s)^\nu &\leq& e_1 \leq (C-p_w)^\nu,
\end{IEEEeqnarray*}
where $B < p_w$.
This implies a value $r^*_{-1} \in (0,1]$ such that $e_1 \leq e_0$, and $r^*<1$.

If $e_1 < e_0$, then $r^*=0$.
On the other hand, if $e_1 = e_0$, as $B<p_w$, it must be the case that $r^*_{-1} \in (0,1]$, which implies $r^* \in [0,1)$.
Consequently, if $B < p_w$, there indeed exists a $(0,0)$-adversary $\mathcal{A}$ and a Nash equilibrium, where each non-adversarial node independently decides to not post its clue to the contract at the slot after a query is received by the contract, with probability $r^*>0$.

Finally, suppose there exists a subgame perfect equilibrium of the dynamic game, where none of the nodes send a valid clue to the client $\client$ over the network by slot $3$, and $\client$ sends its query to the contract by slot $2$.
Once the query appears in the contract, each nodes chooses to withhold its clue from the contract in the next slot with probability $r^*$.

As shown above, if there is a query in the contract, none of the nodes can increase its expected payoff by deviating from the specified action given the other nodes' actions.
Similarly, $\client$ cannot increase its expected payoff by not sending a query to the contract, when none of the nodes sends a valid clue over the network by slot $3$.
Finally, none of the nodes can increase its expected payoff by sending a valid clue to $\client$ over the network by slot $2$, since this does not affect $\client$'s behavior.
This is because $k>1$ and all other nodes refuse to send their clues over the network.
Hence, the claimed action profile indeed constitutes a subgame perfect equilibrium.
However, in this case, there are less than $k$ valid clues in the contract by slot $4$ with probability at least $(1-r^*)^{N-f}>0$.
Thus, there exists a $(0,0)$-adversary and a subgame perfect equilibrium, where $4$-security is violated with positive probability.
\end{proof}

\begin{proof}[Proof of Theorem~\ref{thm:security-two-risk-averse}]
Consider the subgame, where a query is received by the contract of Section~\ref{sec:contract} at some slot $t$.
Let $q^{\mathcal{A}}_{\mathrm{v}}$ denote the probability that given an adversary $\mathcal{A}$, there are less than $k$ valid clues in the contract at slot $t+1$ in the Nash equilibrium with the largest probability of failure.
We first show that there exists a function $q^*(x): [0,(N-f-k+1)(p_s-p_w)) \xrightarrow{} (0,1)$ such that given the slashing function of Section~\ref{sec:contract}, for any $p_0$-adversary $\mathcal{A}$, $0 \leq p_0 < (N-f-k+1)(p_s-p_w)$, it holds that $q^{\mathcal{A}}_{\mathrm{v}} \leq q^*(p_0)$.
By Remark~\ref{remark:strong-adversary}, if $p_0 \geq (N-f-k+1)(p_s-p_w)$, $q^{\mathcal{A}}_{\mathrm{v}}=1$ for any compliant slashing function, including the function of Section~\ref{sec:contract}.

Given a $p_0$-adversary with $p_0<(N-f-k+1)(p_s-p_w)$, consider the Nash equilibrium with the maximum failure probability. 
Let $q_i$ denote the probability that in the equilibrium, $\node_i$, $i \in [N-f]$, does not take the action $\place$, \ie, does not post its clue to the contract, and there are less than $k$ valid clues in the contract at slot $t+1$.
Let $p^i_b$ denote the bribe offered to $\node_i$ by $\mathcal{A}$.
As $\node_i$ maximizes its utility given all other nodes' actions in the equilibrium, any $q_i \in [0,1]$ satisfies the following inequality; otherwise, $\node_i$ can increase its utility by posting its clue to the contract at slot $t+1$:
\begin{IEEEeqnarray*}{C} 
(C-p_w)^\nu \leq q_i(p^i_b + C-p_s)^\nu + (1-q_i)(p^i_b + C-p_w)^\nu,
\end{IEEEeqnarray*}
which further implies
\begin{IEEEeqnarray*}{C} 
q_i \leq \frac{(p^i_b + C-p_w)^\nu-(C-p_w)^\nu}{(p^i_b + C-p_w)^\nu-(p^i_b + C-p_s)^\nu},
\end{IEEEeqnarray*}
where $C=p_s+p_w$.

Let $\mathcal{G}$ denote the set of subsets of $[N-f]$ with $N-f-k+1$ elements.
Let $E_G$, $G \in \mathcal{G}$, denote the event that the nodes in $G$ do not take the action $\place$ at slot $t+1$.
By definition of $q_i$,
\begin{IEEEeqnarray*}{C} 
q_i = \Pr[\cup_{G \in \mathcal{G}: i \in G} E_G] = \Pr[\cup_{G \in G_i} E_G],
\end{IEEEeqnarray*}
where $G_i = \{G \in \mathcal{G}: i \in G\}$.
Similarly, by definition of $q^{\mathcal{A}}_{\mathrm{v}}$,
\begin{IEEEeqnarray*}{C} 
q^{\mathcal{A}}_{\mathrm{v}} = \Pr[\cup_{G \in \mathcal{G}} E_G].
\end{IEEEeqnarray*}
Hence, the function $q^*(p_0)$ is upper bounded by the solution $\tilde{q}^*(p_0)$ to the following optimization problem:
\begin{equation*}
\begin{aligned}
\max_{G \in \mathcal{G}} &\ \Pr[\cup_{G \in \mathcal{G}} E_G]\\
\textrm{s.t.} &\ \Pr[\cup_{G \in G_i} E_G] = q_i \leq \frac{(p^i_b + C-p_w)^\nu-(C-p_w)^\nu}{(p^i_b + C-p_w)^\nu-(p^i_b + C-p_s)^\nu} \quad \forall i \in [N-f] \\
&\ \sum_{i=1}^{N-f} p^i_b \leq p_0 \\
&\ p^i_b \geq 0 \quad \forall i \in [N-f] \\
\end{aligned}
\end{equation*}
Let $\{\tilde{E}_G: G \in \mathcal{G}\}$, denote one set of events for which the value of $\Pr[\cup_{G \in \mathcal{G}} E_G]$ is maximized.
Let $\tilde{p}^i_b$ and $\tilde{q}_i$, $i \in [N-f]$, denote the optimal values of the parameters $p^i_b$ and $q_i$, $i \in [N-f]$, associated with this set of events.

We next construct a $p_0$-adversary $\mathcal{A}$ such that for any compliant slashing function, there exists a Nash equilibrium, where $q^{\mathcal{A}}_{\mathrm{v}} \geq \tilde{q}^*(p_0) \geq q^*(p_0)$. 
This would show that the slashing function of Section~\ref{sec:contract} is security optimal:
\begin{itemize}
    \item Bribes offered to the nodes $\node_i$, $i \in [N-f]$ are given by $\tilde{p}^i_b$.
    \item The probability distribution over the actions adopted by the adversarial and corrupt nodes is determined by the events $\tilde{E}_G$, which satisfy the equation $\Pr[\cup_{G \in G_i} \tilde{E}_G] = \tilde{q}_i$.
\end{itemize}
By definition of the optimization problem, it holds that
\begin{IEEEeqnarray}{C}
\label{eq:constraint}
\tilde{q}_i \leq \frac{(p^i_b + C-p_w)^\nu-(C-p_w)^\nu}{(p^i_b + C-p_w)^\nu-(p^i_b + C-p_s)^\nu} \quad \forall i \in [N-f] 
\end{IEEEeqnarray}

Recall that given any adversary $\mathcal{A}$, each Nash equilibrium can be described by the following variables: $\{\tilde{r}^*_i\}_{i \in [N-f]}$ and $\{r^*_i\}_{i \in [N-f]}$.
Here, $\tilde{r}^*_i$ denotes the probability that $\node_i$ accepts the bribe in the equilibrium, whereas $r^*_i$ denotes the probability that $\node_i$ takes the action $\place$, \ie, posts a valid clue to the contract, at slot $t+1$ in the event that it does not accept the bribe.
We allow $r^*_i$ to be undefined if $\node_i$ accepts the bribe, in which case it will take the action $\place$ as dictated by the adversary.

If a node $\node_i$, $i \in [N-f]$, rejects the bribe, its expected payoff becomes
\begin{IEEEeqnarray*}{C} 
\mathbb{E}[(C+f_i(\mathbf{X})-p_w)^\nu|X_i = 1]r^*_i + \mathbb{E}[(C+f_i(\mathbf{X}))^\nu|X_i = 0](1-r^*_i)
\end{IEEEeqnarray*}
Here the expectation is over the actions of the other nodes in the equilibrium.

As the slashing function satisfies symmetry, for any $i,j \in [N-f]$, $\mathbb{E}[(C+f_i(\mathbf{X})-p_w)^\nu|X_i = 1]=\mathbb{E}[(C+f_j(\mathbf{X})-p_w)^\nu|X_j = 1]=e_1$ and $\mathbb{E}[(C+f_i(\mathbf{X}))^\nu|X_i = 0]=\mathbb{E}[(C+f_j(\mathbf{X}))^\nu|X_j = 0]=e_0$.
As $\node_i$ maximizes its utility given all other nodes' actions in the equilibrium, it must be the case that
\begin{itemize}
    \item $r^*_i=1$, if $e_1 > e_0$. In this case, $\node_i$'s expected utility becomes $e_1$.
    \item $r^*_i=0$, if $e_1 < e_0$.
    \item $r^*_i$ can be any value in $[0,1]$, if $e_1 = e_0$.
    In this case, $\node_i$'s expected utility becomes $e_1$.
\end{itemize}
Here, $e_1$ is upper bounded as shown by the following lemma:
\begin{lemma}
For any $i \in [N-f]$, $p^i_b \geq 0$ and $r^*_i \in [0,1]$,
\label{lem:sublemma-inequality-risk-averse}
\begin{IEEEeqnarray*}{rCl} 
&& (C-p_w)^\nu + (1-r^*_i)(\mathbb{E}[(C -p_w + f_i(\mathbf{X})+ p^i_b)^\nu|X_i = 1]-(C+p^i_b-p_w)^\nu) \\
&\geq& \mathbb{E}[(C+f_i(\mathbf{X})-p_w)^\nu|X_i = 1] = e_1
\end{IEEEeqnarray*}
\end{lemma}

\begin{proof}
For any fixed $a,b,c \in \mathbb{R}^+ \cup \{0\}$ such that $a \geq b \geq c$, it holds that $(a-c)^\nu-(b-c)^\nu \geq a^\nu-b^\nu$, which implies $b^\nu-(b-c)^\nu \geq a^\nu-(a-c)^\nu$.
Moreover, for any compliant slashing function, $-p_s \leq f_i(\mathbf{x}) \leq 0$ for all $\mathbf{x} \in \{0,1\}^N$ and all $i \in [N]$.
Since $C=p_s+p_w$, it holds that $C-p_w+f_i(\mathbf{x}) \geq 0$ for all $\mathbf{x} \in \{0,1\}^N$ and all $i \in [N]$.
Then, for any $\mathbf{x} \in \{0,1\}^{N-f}$ and $i \in [N-f]$,
\begin{IEEEeqnarray*}{rCl} 
&& (C-p_w)^\nu - (C-p_w+f_i(\mathbf{x}))^\nu \geq (C-p_w+p^i_b)^\nu - (C-p_w+f_i(\mathbf{x})+p^i_b)^\nu \\
&\implies& (C-p_w)^\nu - ((C-p_w+p^i_b)^\nu - (C-p_w+f_i(\mathbf{x})+p^i_b)^\nu) \geq (C-p_w+f_i(\mathbf{x}))^\nu \\
&\implies& (C-p_w)^\nu - (1-r^*_i)((C-p_w+p^i_b)^\nu - (C-p_w+f_i(\mathbf{x})+p^i_b)^\nu) \\
&&\ \geq (C-p_w+f_i(\mathbf{x}))^\nu,
\end{IEEEeqnarray*}
since $1-r^*_i \in [0,1]$ and $(C-p_w+p^i_b)^\nu - (C-p_w+f_i(\mathbf{x})+p^i_b)^\nu \geq 0$.
Hence, by linearity of expectation,
\begin{IEEEeqnarray*}{rCl} 
&& (C-p_w)^\nu + (1-r^*_i)(\mathbb{E}[(C-p_w + f_i(\mathbf{X})+ p^i_b)^\nu|X_i = 1]-(C+p^i_b-p_w)^\nu) \\
&\geq& \mathbb{E}[(C-p_w+f_i(\mathbf{X}))^\nu|X_i = 1].
\end{IEEEeqnarray*}
\end{proof}

On the other hand, if $e_1>e_0$ and $\node_i$ accepts the bribe, its expected payoff becomes at least
\begin{IEEEeqnarray*}{rCl} 
&& \left(C +p^i_b - p_s\right)^\nu \tilde{q}_i + \mathbb{E}[(C -p_w + f_i(\mathbf{X})+ p^i_b)^\nu|X_i = 1](1-\tilde{q}_i) \\
&\geq& (C-p_w)^\nu + (1-\tilde{q}_i)(\mathbb{E}[(C -p_w + f_i(\mathbf{X})+ p^i_b)^\nu|X_i = 1]-(C+p^i_b-p_w)^\nu) \\
&\geq& e_1.
\end{IEEEeqnarray*}
Here, the first inequality follows from the equation~\eqref{eq:constraint}, and the last inequality follows from Lemma~\ref{lem:sublemma-inequality-risk-averse}.
Hence, if $e_1 \geq e_0$, the expected utility of node $\node_i$ when it accepts the bribe is at least as large as its expected utility when it rejects the bribe.
Thus, for the nodes $\node_i$, $i \in [N-f]$, there does not exist any action that dominates $(\tilde{r}^*_i,r^*_i)=(1,-)$, implying that they accept the adversary's bribe when any bribe offer is made.
In this case, the action profile specified by the adversary constitutes a Nash equilibrium, and in the equilibrium, $q^{\mathcal{A}}_{\mathrm{v}}$ becomes at least $\tilde{q}^*(p_0) \geq q^*(p_0)$.

Finally, if $e_1 < e_0$, then either the nodes $\node_i$, $i \in [N-f]$, all reject the bribe, and set $r^*_i=0$, or they all accept the bribe. 
The first case happens if for the nodes $\node_i$, $i \in [N-f-k+1]$, there does not exist any action that dominates $(\tilde{r}^*_i,r^*_i)=(0,0)$.
Then, in the equilibrium, $q^{\mathcal{A}}_{\mathrm{v}}$ becomes $1$.
In the latter case, $q^{\mathcal{A}}_{\mathrm{v}} \geq q^*(p_0)$ as argued above, thus concluding the proof.
\end{proof}
\section{Probability of security failure for risk-averse nodes}
\label{sec:appendix-risk-averse-bounds}

\begin{theorem}
\label{thm:optimization-problem}
For any $p_0$, $0 \leq p_0 < (N-f-k+1)(p_s-p_w)$, and $\nu \in (0,1]$, $q^*_{p_0,\nu}$ is the solution to the following optimization problem:
Let $\mathcal{G}$ denote the set of subsets of $[N-f]$ with $N-f-k+1$ elements.
Let $x_G$ denote variables indexed by the sets $G \in \mathcal{G}$.
\begin{equation}
\begin{aligned}
\label{eq:new-optimization-formula}
\max_{j \in [|\mathcal{G}|]} &\ \sum_{j=1}^{|\mathcal{G}|} x_G\\
\textrm{s.t.} &\ \sum_{G \in \mathcal{G}: i \in G} x_G \leq \frac{(p^i_b + p_s)^\nu-(p_s)^\nu}{(p^i_b + p_s)^\nu-(p^i_b + p_w)^\nu} \quad \forall i \in [N-f] \\
&\ \sum_{i=1}^{N-f} p^i_b \leq p_0 \\
&\ p^i_b \geq 0 \quad \forall i \in [N-f] \\
&\ x_G \in [0,1] \quad \forall G \in \mathcal{G}
\end{aligned}
\end{equation}
\end{theorem}

By Remark~\ref{remark:strong-adversary}, $q^*_{p_0,\nu}=1$ if $p_0 \geq (N-f-k+1)(p_s-p_w)$.

\begin{proof}[Proof of Theorem~\ref{thm:optimization-problem}]
From the proof of Theorem~\ref{thm:security-two-risk-averse}, we know that $q^*_{p_0,\nu}$ is the solution to the following optimization problem, where $E_G$ is the event that the nodes in $G$ do not post their clues to the contract at slot $t+1$ when a query is received by the contract at slot $t$:
\begin{equation}
\label{eq:og-optimization-formula}
\begin{aligned}
\max_{G \in \mathcal{G}} &\ \Pr[\cup_{G \in \mathcal{G}} E_G]\\
\textrm{s.t.} &\ \Pr[\cup_{G \in \mathcal{G}: i \in G} E_G] = q_i \leq \frac{(p^i_b + C-p_w)^\nu-(C-p_w)^\nu}{(p^i_b + C-p_w)^\nu-(p^i_b + C-p_s)^\nu} \quad \forall i \in [N-f] \\
&\ \sum_{i=1}^{N-f} p^i_b \leq p_0 \\
&\ p^i_b \geq 0 \quad \forall i \in [N-f]
\end{aligned}
\end{equation}

Let $\{G^j\}_{j \in [|\mathcal{G}|]}$ denote a total order across the events $G \in \mathcal{G}$.
Defining the disjoint events $\hat{E}_{G^1} := E_{G^1}$, and $\hat{E}_{G^i} := E_{G^i} / (\cup_{j \in [i-1]} E_{G^j})$ for $i = 2, \ldots, |\mathcal{G}|$, we observe that the solution of the following optimization is at least as large as the solution to the problem presented by Formula~\eqref{eq:og-optimization-formula}, \ie, $q^*_{p_0,\nu}$:
\begin{equation}
\label{eq:intermediate-optimization-formula}
\begin{aligned}
\max_{j \in [|\mathcal{G}|]} &\ \sum_{j=1}^{|\mathcal{G}|} \Pr[\hat{E}_{G^j}]\\
\textrm{s.t.} &\ \sum_{G \in \mathcal{G}: i \in G} \Pr[\hat{E}_G] \leq \frac{(p^i_b + C-p_w)^\nu-(C-p_w)^\nu}{(p^i_b + C-p_w)^\nu-(p^i_b + C-p_s)^\nu} \quad \forall i \in [N-f] \\
&\ \sum_{i=1}^{N-f} p^i_b \leq p_0 \\
&\ p^i_b \geq 0 \quad \forall i \in [N-f] \\
&\ \Pr[\hat{E}_G] \in [0,1] \quad \forall G \in \mathcal{G}
\end{aligned}
\end{equation}
This is because by definition of $\hat{E}_{G^j}$, the optimal values are the same
\begin{equation*}
    \Pr[\cup_{G \in \mathcal{G}} E_G] = \Pr[\cup_{G \in \mathcal{G}} \hat{E}_G] =  \sum_{j=1}^{|\mathcal{G}|} \Pr[\hat{E}_{G^j}]
\end{equation*}
in both formulas whereas the constraints are relaxed in Formula~\eqref{eq:intermediate-optimization-formula}:
\begin{equation*}
    \sum_{G \in \mathcal{G}: i \in G} \Pr[\hat{E}_G] = \Pr[\cup_{G \in \mathcal{G}: i \in G} \hat{E}_G] \leq \Pr[\cup_{G \in \mathcal{G}: i \in G} E_G]
\end{equation*}

Note that by setting all of $E_G$, $G \in \mathcal{G}$, to be disjoint events, which implies $\hat{E}_G = E_G$, we can ensure that $q^*_{p_0,\nu}$ is the same as the solution to Formula~\eqref{eq:intermediate-optimization-formula}.
In this case, using $C=p_s+p_w$ and defining $x_G := \hat{E}_G$, we can re-write Formula~\eqref{eq:og-optimization-formula} as Formula~\eqref{eq:new-optimization-formula}.
\end{proof}

\begin{theorem}
\label{thm:p0-bound}
Suppose $p_0 < (N-f-k+1)(p_s-p_w)$.
Then,
\begin{IEEEeqnarray*}{rCl}
&& \frac{(p_s+p_b)^\nu-(p_s)^\nu}{(p_s+p_b)^\nu-(p_w+p_b)^\nu} \leq q^*_{p_0,\nu}, \text{ and} \\
&& q^*_{p_0,\nu} \leq 
\frac{1}{N-f-k+1}\min\left(\frac{p_0}{p_s-p_w},\frac{(p_s+p_0)^\nu-(p_s)^\nu}{(p_s+p_0)^\nu-(p_w+p_0)^\nu}\right),
\end{IEEEeqnarray*}
where $p_b = \frac{p_0}{N-f-k+1}$.
\end{theorem}

\begin{proof}[Proof of Theorem~\ref{thm:p0-bound}]
We first prove the upper bound.
Recall that the maximum value for $q^*_{p_0,\nu}$ is given by the solution to Formula~\eqref{eq:new-optimization-formula}, where 
\begin{IEEEeqnarray}{rCl} 
q^*_{p_0,\nu} &=& \sum_{G \in \mathcal{G}} x_G = \frac{1}{N-f-k+1} \sum_{i=1}^{N-f} \sum_{G \in \mathcal{G}: i \in G} x_G \\
\label{eq:upper-bound}
&\leq&  \frac{1}{N-f-k+1} \sum_{i=1}^{N-f} \frac{(p^i_b + C-p_w)^\nu-(C-p_w)^\nu}{(p^i_b + C-p_w)^\nu-(p^i_b + C-p_s)^\nu}
\end{IEEEeqnarray}
Since $\sum_{i=1}^{N-f} p^i_b \leq p_0$ and the function
\begin{equation*}
f(x) = \frac{(x + C-p_w)^\nu-(C-p_w)^\nu}{(x + C-p_w)^\nu-(x + C-p_s)^\nu}
\end{equation*}
is convex in $x$, the sum in equation~\eqref{eq:upper-bound} is maximized when one variable, \eg, $p^1_b$ is set to be $p_0$ and the rest becomes $0$.
Thus, given $C=p_s+p_w$, we obtain
\begin{equation*}
q^*_{p_0,\nu} \leq \frac{1}{N-f-k+1} \frac{(p_s+p_0)^\nu-(p_s)^\nu}{(p_s+p_0)^\nu-(p_w+p_0)^\nu}.
\end{equation*}

We next observe that for any $p^i_b$ such that $0 \leq p^i_b \leq p_s-p_w$\footnote{Offering more bribes to a node is a waste of coins for the adversary. 
Without loss of generality, we can assume $p^i_b \leq p_s-p_w$.}, and any $\nu \in (0,1)$, it holds that
\begin{equation*}
\frac{(p_s+p^i_b)^\nu-(p_s)^\nu}{(p_s+p^i_b)^\nu-(p_w+p^i_b)^\nu} \leq \frac{p^i_b}{p_s-p_w}
\end{equation*}
Thus, if we relax the constraints for $\sum_{G \in \mathcal{G}: i \in G} x_G$ by replacing their upper bounds with $p^i_b/(p_s-p_w)$ in Formula~\eqref{eq:new-optimization-formula}, and maximize the objective, which is possible since the problem has now become convex, we obtain
\begin{equation*}
\frac{1}{N-f-k+1}
\frac{p_0}{p_s-p_w}.
\end{equation*}
(Note that this is the solution to the optimization problem when $\nu=1$.)

Finally, we prove the lower bound on $q^*_{p_0,\nu}$.
For this purpose, we consider the $p_0$-adversary $\mathcal{A}$ that offers a bribe of $p^i_b = p_0/((N-f-k+1)(p_s-p_w)) = p_b$ to the nodes $\node_i$, $i \in [N-f-k+1]$, and for the set $Q \subseteq [N-f-k+1] \cup \{N-f+1,\ldots,N\}$, specifies
\begin{IEEEeqnarray*}{rCl} 
\Pr[X_i = 0, i \in Q] &=& \tilde{q} = \frac{(C+p_b-p_w)^\nu - (C-p_w)^\nu}{(C+p_b-p_w)^\nu-(C+p_b-p_s)^\nu}, \\
\Pr[X_i=1, i \in Q] &=& 1-\tilde{q} = 1-\frac{(C+p_b-p_w)^\nu - (C-p_w)^\nu}{(C+p_b-p_w)^\nu-(C+p_b-p_s)^\nu}.
\end{IEEEeqnarray*}
We will show that given $\mathcal{A}$, for any compliant slashing function, there exists a Nash equilibrium such that $q^{\mathcal{A}}_{\mathrm{v}} \geq \tilde{q}$, which would imply $q^*_{p_0,\nu} \geq \tilde{q}$.
Recall the variables $\tilde{r}^*_i$ and $r^*_i$ from the proof of Theorem~\ref{thm:security-two}.
In the equilibrium, either $(\tilde{r}^*_i,r^*_i) = (1,-)$ for all $i \in [N-f-k+1]$, or $(\tilde{r}^*_i,r^*_i) = (0,0)$ for all $i \in [N-f-k+1]$.
In other words, each node that is offered a bribe either accepts the bribe and becomes corrupted, or does not send a valid clue to the contract at slot $t+1$.

If a node $\node_i$, $i \in [N-f-k+1]$, rejects the bribe, its expected payoff becomes
\begin{IEEEeqnarray*}{C} 
\mathbb{E}[(C+f_i(\mathbf{X})-p_w)^\nu|X_i = 1]r^*_i + \mathbb{E}[(C+f_i(\mathbf{X}))^\nu|X_i = 0](1-r^*_i)
\end{IEEEeqnarray*}
As the slashing function satisfies symmetry, for any $i,j \in [N-f-k+1]$, $\mathbb{E}[(C+f_i(\mathbf{X})-p_w)^\nu|X_i = 1]=\mathbb{E}[(C+f_j(\mathbf{X})-p_w)^\nu|X_j = 1]=e_1$ and $\mathbb{E}[(C+f_i(\mathbf{X}))^\nu|X_i = 0]=\mathbb{E}[(C+f_j(\mathbf{X}))^\nu|X_j = 0]=e_0$.
As $\node_i$ maximizes its utility given all other nodes' actions in the equilibrium, it must be the case that
\begin{itemize}
    \item $r^*_i=1$, if $e_1 > e_0$. In this case, $\node_i$'s expected utility becomes $e_1$.
    \item $r^*_i=0$, if $e_1 < e_0$.
    \item $r^*_i$ can be any value in $[0,1]$, if $e_1 = e_0$.
    In this case, $\node_i$'s expected utility becomes $e_1$.
\end{itemize}

On the other hand, if $e_1>e_0$ and $\node_i$ accepts the bribe, its expected payoff becomes at least
\begin{IEEEeqnarray*}{rCl} 
&& \left(C +p_b - p_s\right)^\nu \tilde{q} + \mathbb{E}[(C + p_b + f_i(\mathbf{X})-p_w)^\nu|X_i = 1](1-\tilde{q}) \\
&\geq& (C-p_w)^\nu + (1-\tilde{q})(\mathbb{E}[(C + p_b + f_i(\mathbf{X})-p_w)^\nu|X_i = 1]-(C+p_b-p_w)^\nu) \\
&\geq& e_1,
\end{IEEEeqnarray*}
where the first equality follows from the value of $\tilde{q}$ specified by the adversary, and the last inequality follows from Lemma~\ref{lem:sublemma-inequality-risk-averse} of Theorem~\ref{thm:security-two-risk-averse}.
Thus, if $e_1 \geq e_0$, then the expected payoff of node $\node_i$ when it accepts the bribe is at least as large as its expected payoff when it rejects the bribe.
Hence, for the nodes $\node_i$, $i \in [N-f-k+1]$, there does not exist any action that dominates $(\tilde{r}^*_i,r^*_i)=(1,-)$.
Then, the action profile specified by the adversary constitutes a Nash equilibrium, and in the equilibrium, $q^{\mathcal{A}}_{\mathrm{v}}$ becomes at least $\tilde{q}$:
\begin{IEEEeqnarray*}{C}
\frac{(C+p_b-p_w)^\nu - (C-p_w)^\nu}{(C+p_b-p_w)^\nu-(C+p_b-p_s)^\nu} = \frac{(p_b+p_s)^\nu - (p_s)^\nu}{(p_b+p_s)^\nu-(p_b+p_w)^\nu},
\end{IEEEeqnarray*}
where $C=p_s+p_w$.

Finally, if $e_1 < e_0$, then either the nodes $\node_i$, $i \in [N-f]$, all reject the bribe, and set $r^*_i=0$, or they all accept the bribe. 
The first case happens if for the nodes $\node_i$, $i \in [N-f-k+1]$, there does not exist any action that dominates $(\tilde{r}^*_i,r^*_i)=(0,0)$.
Then, in the equilibrium, $q^{\mathcal{A}}_{\mathrm{v}}$ becomes $1$.
In the latter case, $q^{\mathcal{A}}_{\mathrm{v}} \geq \tilde{q}$ as argued above, thus concluding the proof.
\end{proof}

\paragraph{Lower and upper bounds on $p_0$ for risk-averse nodes}

\begin{figure}
    \centering
    \includegraphics[width=0.9\linewidth]{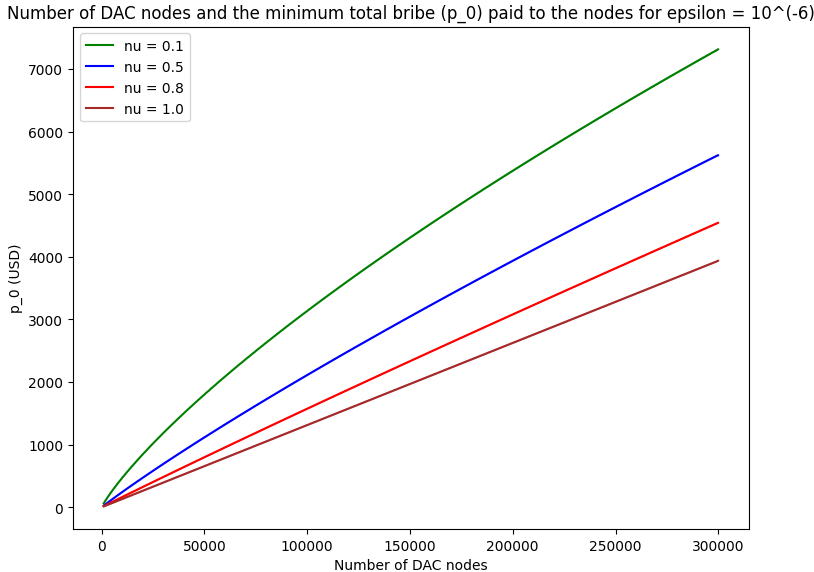}
    \caption{Same as Fig.~\ref{fig:bribe-vs-n}. Lower bounds on the total bribe, $p_0$, needed to ensure that the probability the client does not obtain the answer to its query is $\epsilon = 10^{-6}$, as a function of the number of DAC nodes $N$, and the utility functions $U(x)=x^\nu$, $\nu = 0.1, 0.5, 0.8, 1.0$.}
    \label{fig:bribe-vs-n-2}
\end{figure}

Let $\tilde{p}_0 = p_0 / p_w$ and $\tilde{p}_s = p_s / p_w \approx 1416$ for the given values of $p_s=32$ ETH and $p_w=0.0226$ ETH.
Using the formula of Theorem~\ref{thm:p0-bound}, we can bound the value of $\tilde{p}_0$ as a function of $N$, $\nu$ and $\epsilon$:
$\tilde{p}_{0,\min} \leq \tilde{p_0} \leq \tilde{p}_{0,\max}$,
where $\tilde{p}_{0,\max}$ and $\tilde{p}_{0,\min}$ satisfy the following expressions respectively:
\begin{IEEEeqnarray*}{rCl}
\epsilon &=& \frac{(1416+3\frac{\tilde{p}_{0,\max}}{N})^\nu-(1416)^\nu}{(1416+3\frac{\tilde{p}_{0,\max}}{N})^\nu-(1+3\frac{\tilde{p}_{0,\max}}{N})^\nu} \\
\epsilon &=&
\frac{3}{N}\min\left(\frac{\tilde{p}_{0,\min}}{1415},\frac{(1416+\tilde{p}_{0,\min})^\nu-(1416)^\nu}{(1416+\tilde{p}_{0,\min})^\nu-(1+\tilde{p}_{0,\min})^\nu}\right) \\
&=& \frac{3}{N}\frac{(1416+\tilde{p}_{0,\min})^\nu-(1416)^\nu}{(1416+\tilde{p}_{0,\min})^\nu-(1+\tilde{p}_{0,\min})^\nu}\ \text{for } N<300,000
\end{IEEEeqnarray*}
Solving for $\tilde{p}_{0,\min}$ and $\tilde{p}_{0,\max}$ at different values of $N$, $\nu$ and $\epsilon$, we can calculate the lower and upper bounds on $p_0$ for different $N<300,000$, utility functions of the form $U(x)=x^\nu$ and $\epsilon$.
These bounds are presented by Table~\ref{tab:bribe-bounds} and Figure~\ref{fig:bribe-vs-n}.
We observe that when $p_0\nu/p_s<<1$, the upper and lower bounds differ by at most a constant factor for all values of $\nu$.
\section{Proofs of Theorems~\ref{thm:security-one},~\ref{thm:security-three} and~\ref{thm:security-contract-not-used}}
\label{sec:appendix-proofs-dynamic-game}

\begin{proof}[Proof of Theorem~\ref{thm:security-one}]
We first show that for any $(p_0,p_1)$-adversary, $q(p_0,\nu) \leq q^*_{p_0,\nu}$ for the contract of Section~\ref{sec:contract}.
Consider a game played among the client and the DAC nodes.
At the beginning of slot $2$, the client $\client$ either received $k$ or more valid clues over the network from the nodes, or it did not.
Recall the definitions of the events $Q_{\geq k}$, $Q_{< k}$, $R$ and $P$ from the proof of Lemma~\ref{lem:security-zero}.
Recall Table~\ref{tab:client-utility} summarizing $\client$'s payoff under different actions taken by $\client$.

Since $4$-security is violated only in the event $\overline{R} \land \overline{P}$, given the slashing function of Section~\ref{sec:contract}, $q(p_0,\nu) = \Pr[\overline{R} \land \overline{P}]$.
If $\Pr[\overline{R} \land \overline{P}]=0$ in all Nash equilibria, it trivially holds that $0 = q(p_0,\nu) \leq q^*_{p_0,\nu}$.
Thus, we next assume that there exists a Nash equilibrium where $\Pr[\overline{R} \land \overline{P}]>0$.

Given that the event $R$ happens, \ie, $k$ or more nodes send clues to $\client$ over the network by slot $3$, the adversary cannot distinguish between the actions $Q_{\geq k} \land Q_{< k}$ and $Q_{\geq k} \land \overline{Q}_{< k}$, and between the actions $\overline{Q}_{\geq k} \land Q_{< k}$ and $\overline{Q}_{\geq k} \land \overline{Q}_{< k}$.
Moreover, the utility $(p_f-p_c)^\nu$ of the events $Q_{\geq k} \land Q_{< k}$ and $\overline{Q}_{\geq k} \land Q_{< k}$ exceeds the maximum utility $(p_\mathrm{comp}-p_c)^\nu$ of the events $Q_{\geq k} \land \overline{Q}_{< k}$ and $\overline{Q}_{\geq k} \land \overline{Q}_{< k}$ with the bribe.
Thus, even if the adversary offers an additional payoff of $p_1 < p_{\mathrm{comp}}-p_c$ for the actions $Q_{\geq k} \land \overline{Q}_{< k}$ or $\overline{Q}_{\geq k} \land \overline{Q}_{< k}$, if $\Pr[\overline{R} \land \overline{P}]>0$, then $\client$ can increase its expected utility by reducing the probability of taking $Q_{\geq k} \land \overline{Q}_{< k}$ in favor of $Q_{\geq k} \land Q_{< k}$, and reducing the probability of taking $\overline{Q}_{\geq k} \land \overline{Q}_{< k}$ in favor of $\overline{Q}_{\geq k} \land Q_{< k}$.
Hence, $\client$ never takes the actions $Q_{\geq k} \land \overline{Q}_{< k}$ and $\overline{Q}_{\geq k} \land \overline{Q}_{< k}$ with positive probability in any equilibrium where $\Pr[\overline{R} \land \overline{P}]>0$.
In other words, if $\client$ does not receive $k$ or more clues over the network by slot $3$, it sends a query to the contract by slot $2$ in any equilibrium with a positive failure probability for $4$-security. 

Suppose a query appears in the contract of Section~\ref{sec:contract} at some slot $t$.
By the optimality of the contract, for any given $p_0$ and $\nu$, no compliant contract with a different slashing function can ensure that there are less than $k$ valid clues in the contract at slot $t+1$ with probability less than $q^*_{p_0,\nu}$ in the equilibrium with the maximum failure probability.

Finally, if the event $R$ happens, \ie, the nodes send $k$ or more clues over the network by slot $3$, then $4$-security is satisfied with probability $1$.
If $R$ does not happen, then the client sends a query to the contract by slot $2$, in which case there are $k$ or more clues in the contract by slot $4$ except with probability $q^*_{p_0,\nu}$ in the equilibrium with the maximum failure probability.
Consequently, $4$-security is violated with probability at most $q^*_{p_0,\nu}$, implying that $q(p_0,\nu) \leq q^*_{p_0,\nu}$ for the contract of Section~\ref{sec:contract}.

For the achievability claim, consider a compliant contract and the adversary $\mathcal{A}$ from Theorem~\ref{thm:security-two-risk-averse}.
Suppose there exists a subgame perfect equilibrium, where none of the nodes sends a valid clue over the network by slot $3$, and $\client$ takes the action $\query$, \ie, sends its query to the contract, by slot $2$.
Once the query appears in the contract, the nodes follow the actions specified by $\mathcal{A}$ in the proof of Theorem~\ref{thm:security-two-risk-averse}.

Observe that if there is a query in the contract, none of the nodes can increase its expected utility by deviating from the action specified by the adversary $\mathcal{A}$ given the other nodes' actions.
Similarly, $\client$ cannot increase its expected utility by taking the action $\neg\query$ when none of the nodes sends a valid clue over the network by slot $3$.
Finally, none of the nodes can increase its expected utility by sending a valid clue to $\client$ over the network by slot $2$, since this does not affect $\client$'s behavior.
This is because $k>1$ and all other nodes withhold their clues.
Hence, the claimed action profile indeed constitutes a subgame perfect equilibrium.
However, in this case, the $p_0$-adversary $\mathcal{A}$ ensures that less than $k$ valid clues are sent to the contract by slot $4$ with probability at least $q^*_{p_0,\nu}$ for any compliant contract.
Consequently, there exists a $(p_0,0)$-adversary and a subgame perfect equilibrium such that $4$-security is violated with probability $q^*_{p_0,\nu}$.
\end{proof}

\begin{proof}[Proof of Theorem~\ref{thm:security-three}]
We first augment the adversary $\mathcal{A}$ from Theorem~\ref{thm:security-two-risk-averse} to also offer a bribe of $p_1$ to the client $\client$ such that $p_1 \leq p_{\mathrm{comp}}-p_c$ and $p_1$ satisfies formula~\eqref{eq:p1-constraint}.
In return, $\mathcal{A}$ asks $\client$ to take the action $\query$, \ie to send its query to the contract, by slot $2$ in addition to the actions specified for the nodes.
We show that the action profile, where $\client$ and the nodes accept their bribes, and none of the nodes sends valid clues over the network to $\client$ by slot $3$, constitutes a subgame perfect equilibrium, where security is violated with probability at least $q^*_{p_0,\nu}$.

Suppose a query appears in the contract at some slot $t$.
By the optimality of the slashing function of Section~\ref{sec:contract}, for any given $p_0$ and $\nu$, no contract with a different slashing function can ensure that there are less than $k$ valid clues in the contract at slot $t+1$ with probability less than $q^*_{p_0,\nu}$ in the Nash equilibrium with the maximum failure probability.

Suppose $\client$ receives no valid clues over the network by slot $3$.
Then, if $\client$ does not take the action $\query$ by slot $2$, $\client$'s expected utility becomes at most $0$.
Conversely, if $\client$ takes the action $\query$ by slot $2$, it obtains a utility of at least $(p_\mathrm{comp}-p_c+p_1)^\nu$.
Hence, if none of the nodes sends a valid clue to $\client$ by slot $3$, since $(p_\mathrm{comp}-p_c+p_1)^\nu>0$, $\client$ chooses to take the action $\query$ by slot $2$.

On the other hand, consider an action profile where one or more of the nodes sends a valid clue to $\client$ over the network by slot $2$.
In this case, if $\client$ accepts the bribe and takes the action $\query$ by slot $2$, its utility becomes $(p_f-p_c+p_1)^\nu$ with probability at most $1-q^*_{p_0,\nu}$ and $(p_f-p_c+p_1+p_{\mathrm{comp}})^\nu$ with probability at least $q^*_{p_0,\nu}$.
Thus, its expected utility becomes at least $(1-q^*_{p_0,\nu})(p_f-p_c+p_1)^\nu + q^*_{p_0,\nu}(p_f-p_c+p_1+p_{\mathrm{comp}})^\nu$.
However, if $\client$ rejects the bribe, its expected payoff can at most be $(p_f)^\nu$.
Since this is less than $(1-q^*_{p_0,\nu})(p_f-p_c+p_1)^\nu + q^*_{p_0,\nu}(p_f-p_c+p_1+p_{\mathrm{comp}})^\nu$ by formula~\eqref{eq:p1-constraint}, $\client$ cannot increase its expected utility by rejecting the bribe even if one or more of the nodes sends a valid clue over the network by slot $2$.
Hence, regardless of whether the nodes send valid clues to $\client$ by slot $2$ or not, $\client$ sends a query to the contract by slot $2$, implying that the nodes cannot increase their expected utility by deviating from the action profile claimed to be a subgame perfect equilibrium.

Finally, the action profile, where $\client$ accepts the bribe and sends a query to the contract by slot $2$, and the nodes do not send valid clues over the network by slot $3$, indeed constitutes a subgame perfect equilibrium.
However, in this case, the $p_0$-adversary $\mathcal{A}$ ensures that less than $k$ valid clues are sent to the contract by slot $4$ with probability at least $q^*_{p_0,\nu}$ for any compliant contract.
Hence, $4$-security is violated with probability at least $q^*_{p_0,\nu}$.
\end{proof}

\begin{proof}[Proof of Theorem~\ref{thm:security-contract-not-used}]
For the sake of contradiction, suppose there exists a $(p_0,p_1)$-adversary and a subgame perfect equilibrium, where all of the nodes withhold their clues from the client $\client$ with some positive probability.
By the proof of Theorem~\ref{thm:security-one}, if $\client$ receives no valid clues over the network by slot $2$, it takes the action $\query$, \ie sends query to the contract, by slot $2$.
Since $p_0 < (N-f)p_w$, given any distribution of bribes to the nodes, there exists a non-adversarial node $\node$, which receives a bribe less than $p_w$.
Thus, in the subgame, where $\client$ takes the action $\query$ by slot $2$, the maximum expected utility of $\node$ becomes less than $(C)^\nu$ if it takes the action $\place$ at slot $3$, \ie, sends its clue to the contract, as its bribe cannot compensate for the cost $p_w$ of sending a clue to the contract.
Similarly, $\node$'s expected utility becomes less than $(C)^\nu$ if it does not take the action $\place$ as its bribe cannot compensate for the minimum punishment $p_w+\epsilon$ of not sending a clue to the contract.
Thus, in this subgame perfect equilibrium, $\node$'s expected utility is less than $(C)^\nu$.

Given the action profile, where one of the nodes sends a valid clue to $\client$ over the network by slot $2$, thus enabling $\client$ to recover the response to its query, if $\client$ takes the action $\query$, its utility becomes $(p_f-p_c+p_1)^\nu$ with probability at least $1-q^*_{p_0,\nu}$ and $(p_f-p_c+p_1+p_{\mathrm{comp}})^\nu$ with probability at most $q^*_{p_0,\nu}$. 
In this case, $\client$'s expected utility is upper bounded by $(1-q^*_{p_0,\nu})(p_f-p_c+p_1)^\nu + q^*_{p_0,\nu}(p_f-p_c+p_1+p_{\mathrm{comp}})^\nu$.
Conversely, if $\client$ does not take the action $\query$, its payoff becomes $(p_f)^\nu$.
As 
\begin{IEEEeqnarray*}{C}
(1-q^*_{p_0,\nu})(p_f-p_c+p_1)^\nu + q^*_{p_0,\nu}(p_f-p_c+p_1+p_{\mathrm{comp}})^\nu < (p_f)^\nu,
\end{IEEEeqnarray*}
if at least one node sends a valid clue to $\client$ over the network by slot $2$, $\client$ does not take the action $\query$ at any slot.

Finally, if $\node$ deviates from the claimed equilibrium and sends a valid clue to $\client$ over the network by slot $2$, $\client$ does not take the action $\query$ before the game ends.
In this case, $\node$ obtains a payoff of at least $(C)^\nu$.
However, this is larger than $\node$'s claimed equilibrium payoff, implying a contradiction.
Thus, in all subgame perfect equilibria, there is at least one node that sends a valid clue to $\client$ by slot $2$ over the network, in which case $\client$ does not take the action $\query$ before the game ends.
Consequently, all subgame perfect equilibria of this game satisfies $4$-security without the use of the contract.
\end{proof}

\end{document}